\documentclass[11pt]{article}
\usepackage[utf8]{inputenc}
\usepackage{algorithmicx}

\usepackage[margin=1in]{geometry}
\usepackage{mathtools,bm}
\usepackage{amsmath,amsfonts,amssymb,amsthm,hyperref,algorithm,eqparbox,array,ragged2e,stmaryrd,bbm}
\usepackage{algorithmicx}
\usepackage{bbold}
\usepackage[noend]{algpseudocode}

\algnewcommand\algorithmicinput{\textbf{Input: }}
\algnewcommand\Input{\item[\algorithmicinput]}
\algnewcommand\algorithmicoutput{\textbf{Output: }}
\algnewcommand\Output{\item[\algorithmicoutput]}
\algnewcommand{\OneLineIf}[2]{
  \State \algorithmicif\ #1\ \algorithmicthen\ #2}

\usepackage{thm-restate}
\usepackage{enumitem}
\usepackage[capitalise,nameinlink,noabbrev]{cleveref}
\usepackage{tocloft}
\usepackage[nottoc]{tocbibind} 

\usepackage[dvipsnames]{xcolor}
\hypersetup{
	colorlinks=true,
	urlcolor=MidnightBlue,
	linkcolor=MidnightBlue,
	citecolor=MidnightBlue,
	unicode
}
\usepackage{tikz}
\usetikzlibrary{positioning,decorations.pathreplacing,calligraphy, arrows.meta,patterns}
\usepackage{subcaption} 

\usepackage[most]{tcolorbox} 
\usepackage[dvipsnames]{xcolor}
\newtcbtheorem[number within=section,
               crefname={definition}{Definition},
               Crefname={Definition}{Definition}]{cdefinition}{Definition}%
{colback=pink!20,      
 colframe=pink!80!black, 
 fonttitle=\bfseries,   
 coltitle=black,        
 boxrule=0.8pt,         
 arc=3mm,               
 top=6pt, bottom=6pt, left=6pt, right=6pt, 
 enhanced
 }%
{def}

\crefname{cdefinition}{Definition}{Definitions}
\Crefname{cdefinition}{Definition}{Definitions}

\newtcbtheorem[number within=section,
               crefname={problem}{problems},
  Crefname={Problem}{Problems}]
  {cquestion}{Main Open Problem}%
{colback=yellow!20,      
 colframe=yellow!80!black, 
 fonttitle=\bfseries,   
 coltitle=black,        
 boxrule=0.8pt,         
 arc=3mm,               
 top=6pt, bottom=6pt, left=6pt, right=6pt, 
 enhanced}%
{cquestion}
\crefname{cquestion}{problem}{problems}
\Crefname{cquestion}{Problem}{Problems}

\newtheorem{theorem}{Theorem}[section]
\newtheorem{lemma}[theorem]{Lemma}
\newtheorem{proposition}[theorem]{Proposition}
\newtheorem{claim}[theorem]{Claim}
\newtheorem{corollary}[theorem]{Corollary}

\newtheorem{fact}[theorem]{Fact}

\theoremstyle{definition}
\newtheorem{definition}[theorem]{Definition}
\newtheorem{problem}[theorem]{Open Problem}
\newtheorem{mproblem}[theorem]{Major Open Problem}

\theoremstyle{remark}
\newtheorem{remark}[theorem]{Remark}

\makeatletter
\def\th@example{%
  \thm@notefont{}
  \normalfont 
}
\def\th@definition{%
  \thm@notefont{}
  \normalfont 
}
\makeatother

\theoremstyle{example}

\renewcommand{\setminus}{\smallsetminus}

\newcommand{\rk}{\mathrm{rk}}

\DeclareMathOperator{\D}{D}

\newcommand{\x}{\bm{x}}

\renewcommand{\epsilon}{\varepsilon}

\newcommand{\IP}{\mathrm{IP}}
\newcommand{\Disj}{\textsc{Disj}}

\newcommand{\poly}{\mathrm{poly}}

\newcommand{\br}{\mathsf{br}}
\newcommand{\bc}{\mathsf{bc}}

\newcommand{\sbr}{\mathsf{br}_{\pm}}
\newcommand{\spiky}{\mathsf{spr}}
\newcommand{\spar}{\mathsf{spar}}
\newcommand{\sign}{\mathrm{sign}}
\newcommand{\HDone}[1]{\mathrm{HD}_1^{#1}}

\DeclareMathOperator{\relu}{ReLU}
\DeclareMathOperator{\ltf}{LTF}

\DeclareMathOperator{\row}{row}
\DeclareMathOperator{\col}{col}

\newcommand{\VC}{\mathrm{VC}}

\newcommand{\cR}{\mathcal R}

\def\rank{{\mathrm{rank}}}     

\newtheorem{question}[theorem]{Question}              

\newcommand{\R}{\mathbb{R}}
\newcommand{\F}{\mathbb{F}}

\let\oldsum=\sum 

\RenewDocumentCommand{\sum}{e{_^}}{ 
  \vphantom{\oldsum_{[n]}} 
  \mathop{\smash{\oldsum 
    \IfValueT{#1}{_{#1}} 
    \IfValueT{#2}{^{#2}}
  }}
}
\let\oldprod=\prod 

\RenewDocumentCommand{\prod}{e{_^}}{ 
  \vphantom{\oldprod_{.}} 
  \mathop{\smash{\oldprod 
    \IfValueT{#1}{_{#1}} 
    \IfValueT{#2}{^{#2}}
  }}
}

\newcommand{\clsNP}{\textsf{NP}}

\newcommand{\clsPH}{\textsf{PH}}
\newcommand{\clsPSPACE}{\textsf{PSPACE}}

\crefname{section}{Section}{Sections}
\Crefname{cdefinition}{Definition}{Definitions}

\title{Spiky Rank and Its Applications to Rigidity and Circuits}
 \author{Lianna Hambardzumyan \thanks{The University of Copenhagen, Denmark. \texttt{lianna.hambardzumyan@gmail.com}. This work is funded by the European Research Council (ERC) under grant agreement no. 101125652 (ALBA). Most of the work was done while the author was a postdoctoral researcher at the University of Victoria, Canada, funded by NSERC and at Hebrew University of Jerusalem, Israel, funded by ISF grants 921/22 and 2635/19.} 
\and
Konstantin Myasnikov \thanks{EPFL, Switzerland, \texttt{kostyamyasso31@gmail.com}. Supported by Swiss
State Secretariat for Education, Research, and Innovation (SERI) under contract number MB22.00026.}
\and
Artur Riazanov \thanks{EPFL, Switzerland, \texttt{tunyash@gmail.com}. Supported by Swiss
State Secretariat for Education, Research, and Innovation (SERI) under contract number MB22.00026.}
\and
Morgan Shirley \thanks{Lund University, Sweden. \texttt{morgan.shirley@cs.lth.se}. Supported by Knut and Alice Wallenberg grant KAW 2023.0116.}
\and
Adi Shraibman \thanks{School of Computer Science, The Academic College of Tel Aviv-Yaffo, Israel,
\texttt{adish@mta.ac.il}.
}
}
\begin{document}
\maketitle
\begin{abstract}

We introduce spiky rank, a new matrix parameter that enhances blocky rank by combining the combinatorial structure of the latter with linear-algebraic flexibility. A spiky matrix is block-structured with diagonal blocks that are arbitrary rank-one matrices, and the spiky rank of a matrix is the minimum number of such matrices required to express it as a sum. This measure extends blocky rank to real matrices and is more robust for problems with both combinatorial and algebraic character. 

Our conceptual contribution is as follows: we propose spiky rank as a well-behaved candidate matrix complexity measure and demonstrate its potential through applications. We show that large spiky rank implies high matrix rigidity, and that spiky rank lower bounds yield lower bounds for depth-2 ReLU circuits, the basic building blocks of neural networks. On the technical side, we establish tight bounds for random matrices and develop a framework for explicit lower bounds, applying it to Hamming distance matrices and spectral expanders. Finally, we relate spiky rank to other matrix parameters, including blocky rank, sparsity, and the $\gamma_2$-norm.
\end{abstract}
\newgeometry{margin=1in,top=1in,bottom=1in}

\section{Introduction}
Studying the structural parameters of matrices has proven to be a powerful way to capture the complexity of various computational models and to establish lower bounds. A recurring theme in complexity theory is identifying 
``well-behaved'' matrix parameters: parameters strong enough to capture meaningful lower bounds, structured enough to admit combinatorial or algebraic analysis, and flexible enough to connect to multiple computational models. Rank, rigidity, $\gamma_2$-norm, and blocky rank are all examples of such measures. 

In this work we introduce and study \emph{spiky rank}, a new matrix parameter that we argue meets the criteria for a well-behaved complexity measure.
Spiky rank is a generalization of the blocky rank defined for communication complexity applications by Hambardzumyan, Hatami and Hatami \cite{HHH23} (\cref{def:matrices}, \Cref{def:ranks}) which later found broad range of applications in complexity theory and beyond \cite{DY24, pitassiStrength2023, williamsOrthogonal2024, goos2025equality, goh2025block, applebaum2025meta}.

Spiky rank adds a linear algebraic flavor to the combinatorial nature of blocky rank. This allows us to connect spiky rank to complexity measures that are both algebraic and combinatorial in their nature, such as matrix rigidity and circuits with ``algebraic'' gates.

Both blocky rank and spiky rank fit into a common paradigm for defining complexity measures: 
we begin by specifying a class of simple objects that we define to have complexity 1, and then define 
the complexity of a general object as the minimum number of such simple objects needed to construct it. 
To define blocky and spiky rank, we define our matrices of complexity 1 as follows: 

\begin{cdefinition}{Blocky and Spiky Matrices}{matrices} \label{def:blocky_matrix}
    A \textbf{blocky matrix} is a matrix that can be obtained from the identity matrix by permuting rows or columns, duplicating rows or columns, and adding all-zero rows or columns. Equivalently, a blocky matrix is a matrix whose rows and columns can be permuted so that it consists of diagonal blocks that are all-one submatrices (of possibly different sizes), with all off-diagonal blocks equal to zero. See \cref{fig:blocky}.

    \medskip
     A \textbf{spiky matrix} is a blocky matrix in which each diagonal block is an arbitrary rank-one matrix (rather than an all-one matrix). Equivalently, a spiky matrix can be expressed as the entrywise product of a blocky matrix and a rank-one matrix.  See \cref{fig:spiky}.
\end{cdefinition}

\begin{figure}[H]
    \def\vectoru{25,50,25,50,100,75,75,50,75,100,100,75,75,50,50,100}
    \def\vectorv{50,100,75,25,75,100,50,100,50,50,100,25,50,100,75,50}
    \newcommand{\gridsize}{0.25}
    \newcommand{\vectoroffset}{0.4}
    
    \begin{subfigure}[b]{0.3\textwidth}
        \centering
        \begin{tikzpicture}
            \draw[fill=pink!85!black] (0, 16*\gridsize) rectangle (5*\gridsize, 10*\gridsize);
            \draw[fill=pink!85!black] (5*\gridsize, 10*\gridsize) rectangle (6*\gridsize, 9*\gridsize);
            \draw[fill=pink!85!black] (6*\gridsize, 9*\gridsize) rectangle (10*\gridsize, 5*\gridsize);
            \draw[fill=pink!85!black] (10*\gridsize, 5*\gridsize) rectangle (14*\gridsize, 3*\gridsize);
            \draw[fill=pink!85!black] (14*\gridsize, 3*\gridsize) rectangle (16*\gridsize, 0);
            
            \draw (0, 0) rectangle (16*\gridsize, 16*\gridsize);
        \end{tikzpicture}
        \caption{A blocky matrix $B$.}
        \label{fig:blocky}
    \end{subfigure}
    \quad \quad
    \begin{subfigure}[b]{0.3\textwidth}
        \centering
        \begin{tikzpicture}
            \foreach \i[count=\icnt] in \vectoru{
                \fill[pink!85!black!\i!white] ({(\icnt-1)*\gridsize}, {16*\gridsize+\vectoroffset}) rectangle ({(\icnt)*\gridsize}, {16*\gridsize+\gridsize+\vectoroffset});
            }
            \draw (0, {16*\gridsize+\vectoroffset}) rectangle (\gridsize*16, {16*\gridsize+\gridsize+\vectoroffset});
            
            \foreach \j[count=\jcnt] in \vectorv{
                \fill[pink!85!black!\j!white] ({0 - \gridsize - \vectoroffset}, {(\jcnt-1)*\gridsize}) rectangle ({0 - \vectoroffset}, {(\jcnt)*\gridsize});
            }
            \draw ({0 - \gridsize - \vectoroffset}, 0) rectangle ({0 - \vectoroffset}, \gridsize*16);
            
            \foreach \i[count=\icnt] in \vectoru{
                \foreach \j[count=\jcnt,evaluate=\j as \shade using {\i*\j/100}] in \vectorv{
                    \fill[pink!85!black!\shade!white] ({(\icnt-1)*\gridsize}, {(\jcnt-1)*\gridsize}) rectangle ({(\icnt)*\gridsize}, {(\jcnt)*\gridsize});
                }
            }
            \draw (0, 0) rectangle (16*\gridsize, 16*\gridsize);
            
        \end{tikzpicture}
        \caption{A rank-one matrix $u \otimes v$.}
    \end{subfigure}
    \quad \quad
    \begin{subfigure}[b]{0.3\textwidth}
        \centering
        \begin{tikzpicture}
            \foreach \i[count=\icnt] in \vectoru{
                \fill[pink!85!black!\i!white] ({(\icnt-1)*\gridsize}, {16*\gridsize+\vectoroffset}) rectangle ({(\icnt)*\gridsize}, {16*\gridsize+\gridsize+\vectoroffset});
            }
            \draw (0, {16*\gridsize+\vectoroffset}) rectangle (\gridsize*16, {16*\gridsize+\gridsize+\vectoroffset});
            
            \foreach \j[count=\jcnt] in \vectorv{
                \fill[pink!85!black!\j!white] ({0 - \gridsize - \vectoroffset}, {(\jcnt-1)*\gridsize}) rectangle ({0 - \vectoroffset}, {(\jcnt)*\gridsize});
            }
            \draw ({0 - \gridsize - \vectoroffset}, 0) rectangle ({0 - \vectoroffset}, \gridsize*16);
            
            \foreach \i[count=\icnt] in \vectoru{
                \foreach \j[count=\jcnt,evaluate=\j as \shade using {\i*\j/100}] in \vectorv{
                    \fill[pink!85!black!\shade!white] ({(\icnt-1)*\gridsize}, {(\jcnt-1)*\gridsize}) rectangle ({(\icnt)*\gridsize}, {(\jcnt)*\gridsize});
                }
            }
            
            \fill[white] (5*\gridsize, 16*\gridsize) rectangle (16*\gridsize, 10*\gridsize);
            \fill[white] (0, 10*\gridsize) rectangle (5*\gridsize, 9*\gridsize);
            \fill[white] (6*\gridsize, 10*\gridsize) rectangle (16*\gridsize, 9*\gridsize);
            \fill[white] (0, 9*\gridsize) rectangle (6*\gridsize, 5*\gridsize);
            \fill[white] (10*\gridsize, 9*\gridsize) rectangle (16*\gridsize, 5*\gridsize);
            \fill[white] (0, 5*\gridsize) rectangle (10*\gridsize, 3*\gridsize);
            \fill[white] (14*\gridsize, 5*\gridsize) rectangle (16*\gridsize, 3*\gridsize);
            \fill[white] (0, 3*\gridsize) rectangle (14*\gridsize, 0);
            
            \draw (0, 16*\gridsize) rectangle (5*\gridsize, 10*\gridsize);
            \draw (5*\gridsize, 10*\gridsize) rectangle (6*\gridsize, 9*\gridsize);
            \draw (6*\gridsize, 9*\gridsize) rectangle (10*\gridsize, 5*\gridsize);
            \draw (10*\gridsize, 5*\gridsize) rectangle (14*\gridsize, 3*\gridsize);
            \draw (14*\gridsize, 3*\gridsize) rectangle (16*\gridsize, 0);
            
            \draw (0, 0) rectangle (16*\gridsize, 16*\gridsize);
        \end{tikzpicture}
        \caption{A spiky matrix $B \circ (u \otimes v)$.}
        \label{fig:spiky}
    \end{subfigure}
    \caption{A spiky matrix is the entrywise product of a blocky matrix and a rank-one matrix.}
\end{figure}
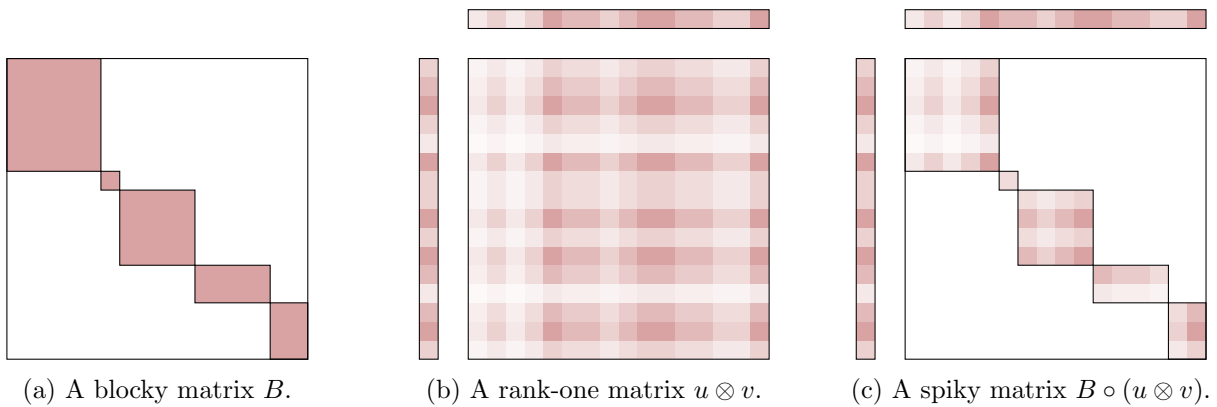

\newpage
We then define the complexity of general matrices as follows:

\begin{cdefinition}{Blocky and Spiky Rank}{ranks}\label{def:blocky_rank}  A blocky matrix is said to have \emph{blocky rank one}.  The \textbf{blocky rank} of a matrix $M$, denoted $\br(M)$, is the minimum number $r$ such that $M$ can be written as linear combination of $r$ blocky rank one matrices.

    \medskip
    A spiky matrix is said to have \emph{spiky rank one}. The \textbf{spiky rank} of a matrix $M$, denoted $\spiky(M)$, is the minimum number $r$ such that $M$ can be written as a sum of $r$ spiky rank one matrices. 
\end{cdefinition}

\subsection{Motivations for Blocky and Spiky Rank}
\paragraph{Blocky rank.} The study of blocky matrices and related decompositions has a surprisingly long and diverse history. 
Although the terminology ``blocky matrix'' is recent, the same objects have been rediscovered many times under different names across theoretical computer science and mathematics.
For instance, fat matchings  \cite{jukna2006graph,PudlakRodl}, equivalence graphs \cite{frankl1982covering, alon1986covering}, adjacency matrices of $P_4$-free graphs \cite{blockP4free2021, applebaum2025meta}, and equality matrices \cite{williamsOrthogonal2024}. Various versions of blocky decompositions have been studied; blocky rank, \emph{blocky covering} (union of blocky matrices), \emph{blocky partitioning} (sum of blocky matrices) and more. Each of these variants has turned out to capture useful phenomena in different research areas:
\begin{itemize}[itemsep=0pt]
    
\item  In \emph{communication complexity} blocky decompositions have been used to study models of communication with access to the equality oracle \cite{HHH23, pitassiStrength2023, blockP4free2021, goos2025equality}.

\item In \emph{operator theory}, it was observed \cite{HHH23} that blocky matrices are the contractive idempotents of Schur multiplier algebras, making blocky rank a natural way to characterize the structure of idempotents in this setting (see also \cite{goh2025block}).

\item In \emph{combinatorics}, blocky coverings of different families of graphs have been studied in \cite{frankl1982covering, alon1986covering, blockP4free2021, DY24}. 

\item In \emph{cryptography}, Applebaum and Nir \cite{applebaum2025meta} studied blocky covering to characterize families of graphs that can be realized with linear secret sharing schemes of constant size. 

\item In \emph{circuit complexity}, blocky decompositions provide lower bounds: Jukna \cite{jukna2006graph} used blocky partitioning to derive bounds against depth-3 circuits with parity gates, and Avraham and Yehudayoff \cite{DY24} showed that blocky partitioning lower bounds depth-2 threshold circuits.  

\item \emph{Branching program} size lower bounds were obtained via blocky covering in \cite{PudlakRodl}. 

\end{itemize}

Perhaps the most striking and far-reaching connection is due to Williams \cite{williamsOrthogonal2024}, who related blocky rank to the renowned Orthogonal Vectors Conjecture (OVC). In particular, Williams showed that to refute OVC -- and consequently the Strong Exponential Time Hypothesis (SETH) -- it would suffice to prove that the $2^n \times 2^n$ Disjointness matrix $D_n$ satisfies $\br(D_n) \leq 2^{o(n)}$. On the other hand, if one could show that $\br(D_n) \geq 2^{\omega(n)}$, this would immediately yield exponential lower bounds for depth-2 exact threshold circuits by \cite{DY24, williamsOrthogonal2024}, resolving a notorious open problem in circuit complexity. This duality highlights the importance of studying blocky rank: either upper bounds or lower bounds on $\br(D_n)$ would have transformative implications in algorithms, fine-grained complexity, and circuit lower bounds.

\paragraph{Spiky rank.} As evident from this wide range of applications, blocky rank is a versatile measure. But it is fundamentally a combinatorial measure and, in particular, not a robust one. Small modifications to the matrix can cause its blocky rank to change drastically.

For example, in \cref{fig:identity-vs-diagonal}, the identity matrix has $\br(I_n)=1$, but if we place distinct weights on the diagonal, the blocky rank jumps to $\br(D_n)=n$. Structurally these matrices are nearly the same -- they are both just diagonal matrices -- but blocky rank treats them in a completely different way. This lack of robustness makes blocky rank poorly suited for applications involving real-valued matrices or settings where algebraic properties matter. Spiky rank addresses this shortcoming by capturing the intuition that a diagonal matrix is expected to be simple regardless of its entries, in particular, $\spiky(D_n)=1$.

\begin{figure}[h]
    \centering
    \[
    I_n =
    \begin{bmatrix}
        1 & 0 & 0 & \cdots & 0 \\
        0 & 1 & 0 & \cdots & 0 \\
        0 & 0 & 1 & \cdots & 0 \\
        \vdots & \vdots & \vdots & \ddots & \vdots \\
        0 & 0 & 0 & \cdots & 1
    \end{bmatrix}
    \qquad
    D_n =
    \begin{bmatrix}
        1 & 0 & 0 & \cdots & 0 \\
        0 & 2 & 0 & \cdots & 0 \\
        0 & 0 & 3 & \cdots & 0 \\
        \vdots & \vdots & \vdots & \ddots & \vdots \\
        0 & 0 & 0 & \cdots & n
    \end{bmatrix}
    \]
    \caption{The identity matrix $I_n$ (left) and the diagonal matrix $D
    _n$ (right).}
    \label{fig:identity-vs-diagonal}
\end{figure}

 This extra flexibility makes spiky rank well-suited for problems that have both combinatorial and algebraic character. In this work, we demonstrate two such applications: 
 \begin{itemize}
     \item Spiky rank lower bounds \emph{matrix rigidity}, a central problem in circuit and communication complexity (\cref{sec:intro-rigidity}),
     \item Spiky rank lower bounds the size of \emph{depth-2 ReLU circuits}, which sit at the intersection of complexity theory and modern machine learning (\cref{sec:intro_circuit}).
 \end{itemize}

Finally, note that spiky rank strictly strengthens blocky rank: every blocky decomposition is also spiky, so $\spiky(M) \leq \br(M)$ for all $M$. Thus, lower bounds for spiky rank have the same implications as for blocky rank. On the other hand, upper bounds can be even more significant. As with blocky rank, upper bounds for spiky rank of the Disjointness matrix have far-reaching consequences: if the $2^n \times 2^n$ Disjointness matrix $\Disj_n$ satisfies $\spiky(\Disj_n) \leq 2^{o(n)}$, then Orthogonal Vector Conjecture (and thus SETH) would be false \cite{williams2024personal_communication}.

\subsection{Results}\label{sec:applications}
Our main contribution is conceptual: we introduce spiky rank as a new matrix complexity measure, demonstrate its applicability to problems in complexity theory, and establish its basic properties. On the technical side, we prove lower bounds for both random and explicit matrix families and relate spiky rank to other classical parameters. Our main results are:

\begin{enumerate}
\item We demonstrate the natural applicability of spiky rank to two problems in
complexity theory: matrix rigidity (\cref{sec:intro-rigidity}) and ReLU
circuits (\cref{sec:intro_circuit}).
\item
We show that most Boolean matrices have spiky rank $\Omega(N / \log N)$ (matching blocky rank), while most real matrices have spiky rank $\Omega(N)$ (\cref{sec:bounds_random}).

\item
We give the following bounds on the spiky rank of specific $N \times N$ matrices of interest:

\begin{itemize}
\item $\Omega(\sqrt{\log N})$ for the 1-Hamming Distance matrix (\cref{sec:hd_bound}), simultaneously improving the previous $\log\log N$ bound for blocky rank as well \cite{HHH23}.
\item $\Omega(\log N)$ for adjacency matrices of certain expander graphs (\cref{sec:expanders_bound}).
\item $\Omega(\log N / \log\log N)$ for the Inner Product and Disjointness matrix (\cref{sec:ip_bound}) (likely far from optimal but it represents a first step toward stronger bounds). 
\end{itemize}

\item For Boolean matrices, we establish a dimension-free relation between blocky and spiky rank, implying in particular that constant spiky rank yields constant blocky rank (\cref{sec:matrix_parameters}). We further bound spiky rank in terms of matrix sparsity (\cref{sec:sparsity}), separate it from the $\gamma_2$-norm (\cref{sec:gamma-2-vs-br}), and highlight gaps between the approximate and sign variants of blocky and spiky rank (\cref{sec:approx_sign}).
\end{enumerate}

\subsubsection{Application to Matrix Rigidity}
\label{sec:intro-rigidity}
Matrix rigidity, introduced by Valiant \cite{valiant1977graph} and Grigoriev \cite{grigoriev}, asks how far a matrix is from having low rank: the rank-$r$ rigidity $\cR_M(r)$ is the minimum number of entry modifications required to reduce the rank of $M$ to $r$.

Highly rigid matrices are central to long-standing open problems in complexity theory. Valiant showed that if $M \in \R^{N \times N}$ satisfies $\cR_M(\epsilon N) \geq \Omega(N^{1+\delta})$ for some $\epsilon,\delta>0$, then any log-depth linear circuit computing the linear transformation $x \mapsto Mx$ must have superlinear size.
Additionally, Razborov \cite{razborovRigidity} linked rigidity to communication complexity, showing that matrices in $\clsPH^{cc}$ -- the communication complexity analogue of the polynomial hierarchy -- cannot be too rigid. In particular, if a $N \times N$ Boolean matrix $M \notin \clsPH^{cc}$, then  its rigidity over $\F_2$ must satisfy \[\cR^{\F_2}_M(2^{(\log\log N)^{\omega(1)}}) \geq \frac{N^2}{2^{(\log\log N)^{O(1)}}}.\] Hence, constructing sufficiently rigid matrices would yield lower bounds for $\clsPH^{cc}$ and ultimately separate it from $\clsPSPACE^{cc}$ -- both of which remain major open problems in communication complexity.

Despite extensive study, progress has been limited on finding explicit rigid matrices: the best explicit constructions fall well short of the parameters needed for Valiant’s or Razborov’s applications.  

A direct link between rigidity and spiky rank is rather straightforward to show: matrices with large spiky rank are also highly rigid. 

\begin{restatable}{theorem}{sprlbonrigidity}\label{thm:spr_lower_bound_on_rigidity}
    Let $M$ be a matrix and $0 < r \leq \spiky(M)$. Then, 
    $$\cR_M(r) \;\geq\; \frac{(\spiky(M) - r)^2}{4}.$$
\end{restatable}
This has two key consequences. First, for Valiant's target rank $\epsilon N$, an explicit matrix with $\spiky(M) = \Omega(N)$ would achieve the strongest possible rigidity bound $\Omega(N^2)$. Note that, since Boolean matrices have spiky rank at most ${N/\log N}$, the only candidate matrices with spiky rank $\Omega(N)$ are the real matrices.

Second, for Razborov's target rank $2^{(\log\log N)^{\omega(1)}}$, we would need to find an explicit matrix with  spiky rank at least (see \cref{cor:spiky_ph}) \[\frac{N}{2^{(\log\log N)^{O(1)}}}.\]

This shows that spiky rank is not only a structural parameter, but also a potential tool for constructing explicit rigid matrices. With this goal in mind, we highlight the following open challenges to guide future research in this direction. 

\begin{mproblem}
    Find an explicit real matrix $M \in \R^{N \times N}$ with $\spiky(M)=\Omega(N)$.
\end{mproblem}

\begin{mproblem}
    Find an explicit Boolean matrix $M \in \{0,1\}^{N \times N}$ with $$\spiky(M)=\frac{N}{2^{(\log\log N)^{O(1)}}}.$$
\end{mproblem}

Although most Boolean matrices have spiky rank at least $\Omega(N/\log N)$ and most real matrices have $\Omega(N)$ (see \cref{sec:bounds_random}), constructing explicit examples that meet these bounds remains elusive. We expect that solving either problem above will be quite challenging. At present, the best lower bounds we can prove for explicit matrices are of order $\Omega(\log N)$ (see \cref{sec:bounds_specific}).

Finally, we remark that \cref{thm:spr_lower_bound_on_rigidity} allows known non-rigidity results to translate into upper bounds on spiky rank. For instance, Alman and Williams~\cite{AW17} proved that for sufficiently small $\varepsilon > 0$, the Inner Product matrix $\IP_n$ of size $2^n \times 2^n$ satisfies $$\cR_{\IP_n}(2^{n(1-f(\varepsilon))}) \leq 2^{n(1+\varepsilon)},$$ where $f(\varepsilon) = \Theta(\varepsilon^2 / \log(1/\varepsilon)).$
By \cref{thm:spr_lower_bound_on_rigidity}, this yields an upper bound $\spiky(\IP_n) \leq 2^{(1-\delta)n}$ for some constant $\delta > 0$, improving on the trivial bound of ${2^n}/{n}$ (see \cref{sec:bounds_random}).

\subsubsection{Application to Circuit Complexity}\label{sec:intro_circuit}

Connections between blocky rank and circuit complexity were first observed by Jukna \cite{jukna2006graph}, and later by Avraham and Yehudayoff \cite{DY24}, who showed that the blocky rank of a matrix lower bounds the size of any $\Sigma \circ \ltf$ circuit computing it up to a certain loss factor (implicit in \cite{hansen2010exact}). 
Here, $\Sigma \circ \ltf$ denotes linear combinations of \emph{linear threshold functions} (LTFs), where an LTF is a Boolean function of the form $T(x) = \mathbb{1}[\langle w, x \rangle \geq \alpha]$ for some $w \in \{0,1\}^n$ and $\alpha \in \R$. 

Naturally this connection extends to spiky rank and to $\Sigma \circ \relu$ circuits, where ReLU functions generalize threshold functions by outputting $\max\{0, \langle w, x \rangle - \alpha\}$. The following proposition admits a simple proof, yet it is conceptually important.
\begin{proposition}
For $M\colon \{0,1\}^n \times \{0,1\}^n \to \R$, let $s$ be the size of a $\Sigma \circ \relu$ \-circuit that computes $M$. Then 
$$s \geq \frac{\spiky(M)}{3(n+1)}.$$
\end{proposition}

Observe that $\Sigma \circ \mathrm{LTF}$ circuits can be simulated by $\Sigma \circ \relu$ circuits with only constant factor blowup in size, since
\[\mathbb{1}[\langle w, x \rangle \geq \alpha] = \max\{0, \langle w, x \rangle - \alpha+1\} - \max\{0, \langle w, x \rangle - \alpha\}. \]

Linear threshold and ReLU circuits are central models in both complexity theory and machine learning.
$\Sigma \circ \mathrm{LTF}$ circuits capture linear combinations of halfspaces and have been studied since the early days of threshold logic. On the other hand, ReLU gates have become the dominant activation function in modern neural networks, thus $\Sigma \circ \mathrm{ReLU}$ circuits are the basic building blocks of modern neural networks.  However, our understanding of their expressive power from a Boolean complexity perspective remains limited.  Understanding the power and limitations of these models is therefore both an important task in circuit complexity and in machine learning. 

Despite decades of work, the best known lower bounds for these circuits are very weak: even for $\Sigma \circ \ltf$, essentially no superlinear lower bounds are known for explicit functions \cite{roychowdhury1994lower}. 
Williams \cite{williams2018limits} gave the strongest results to date, proving that for every $k$ there exists a function in $\clsNP$ that requires sparsity $n^k$ in both $\Sigma \circ \ltf$ and $\Sigma \circ \relu$ circuits. For explicit functions, the strongest known result is due to Kane and Williams \cite{Kane_Williams2016}, who showed that even the stronger class $\ltf \circ \ltf$ requires size $\Omega(n^{3/2})$ to compute the $n$-variate Andreev function. Mukherjee and Basu \cite{mukherjee2017lower} extend this to prove a lower bound of $\Omega(n^{1-\delta})$ for $\ltf \circ \relu$ circuits computing the Andreev function. 

Our theorem shows that spiky rank serves as a lower bound measure for depth-two neural networks with ReLU activation, just as blocky rank does for $\Sigma \circ \ltf$ circuits. In particular, proving spiky rank lower bounds of $\omega((\log N)^{5/2})$ for explicit $N \times N$ matrices would immediately yield new circuit lower bounds.  
At present, our lower bounds are of order $\log N$; the first milestone toward this goal would be to establish \emph{any} super-logarithmic spiky rank lower bounds for explicit matrices.

\begin{problem}
       Find an explicit real matrix $M \in \R^{N \times N}$ with $\spiky(M)=\omega((\log N)^{5/2})$. 
\end{problem}


 \subsubsection{Bounds for Random Matrices}\label{sec:bounds_random}
 How large can the spiky rank of a matrix be? 
Trivially, both the spiky rank and blocky rank of an $N \times N$ Boolean matrix are at most $N$. 
For blocky rank, this bound can be improved to $\tfrac{N}{\log N}$ \cite{PudlakRodl,DY24}, and therefore the same upper bound applies to spiky rank as well. 
A natural question is whether spiky rank admits a substantially stronger upper bound. 
Avraham and Yehudayoff \cite{DY24} showed that most matrices are maximally hard for blocky rank. 
We prove that the same phenomenon holds for spiky rank.
\begin{theorem}[Informal]
A random Boolean matrix of size $N \times N$ has spiky rank of order $\Omega\big(\tfrac{N}{\log N}\big)$ with high probability. 
\end{theorem}
For real matrices, the situation is more extreme: the trivial upper bound of $N$ turns out to be essentially tight.  

\begin{theorem}[Informal]
    A random real matrix of size $N \times N$ has spiky rank at least $\tfrac{N}{2}$ with high probability.
\end{theorem}
This mirrors what is known for many other natural measures in complexity theory, such as circuit complexity and communication complexity, where random functions or matrices almost always exhibit the largest possible complexity.

\subsubsection{Bounds for Specific Matrices.}\label{sec:bounds_specific}
The lower bounds discussed above, while strong, are for random (non-explicit) matrices. 
For applications in complexity theory, however, it is essential to establish lower bounds for explicit families of matrices. 

For blocky rank, the best known lower bounds for explicit families of matrices are of order $\Omega(\log N)$ (Greater-Than \cite{DY24}, Inner-Product, Disjointness \cite{HHH23}).

For spiky rank, we develop a general lower bound framework (\cref{thm:sprlbframe}) and apply it to two examples: adjacency matrices of the Hamming cube and expander graphs. The common properties of these graphs that we use are:
\begin{description}[noitemsep]
    \item[Thinness:] every small subgraph has smaller edge-to-node ratio compared to the entire graph; 
    \item[Induced matchings:] there are large induced matchings in all large subgraphs.
\end{description}
This framework can only yield lower bounds up to $\Theta(\log N)$, the bottleneck is in the use of a variation of a lemma from \cite[Lemma~12]{AvrahamYehudayoff2022}.

\medskip
\noindent
\emph{Hamming distance.}  Consider the matrix $\HDone{n}$ of size $N \times N$ with $N = 2^n$, whose rows and columns are indexed by all $n$-bit strings, and where  
\[
\HDone{n}(x,y) = 
\begin{cases}
1 & \text{if $x$ and $y$ differ in exactly $1$ coordinate,}\\
0 & \text{otherwise.}
\end{cases}
\]

Applying our framework to this matrix yields the following bound:  

\begin{restatable}{theorem}{lbforhd} 
\label{thm:lower-bound-hd1}
    $\spiky(\HDone{n}) \geq \Omega(\sqrt{\log N}).$
\end{restatable}
This bound is not far from the optimal; upper bound on the spiky rank is $\spiky(\HDone{n}) \leq \br(\HDone{n})\leq \log N$, where the last upper bound follows from the following blocky matrix decomposition: Let $M_i$ be the matrix such that $M_i(x,y) = 1$ iff $x_i \neq y_i$ and all the other entries of $x$ and $y$ coincide. Then $M_i$ is a permutation matrix, so $\br(M_i) = 1$ and $\HDone{n} = M_1 + \dots + M_n$.

This lower bound is quite remarkable considering that $\HDone{n}$ matrix is very simple with respect to many other matrix complexity measures, as we discuss in \cref{sec:matrix_parameters}.
\begin{problem}
   Is $\spiky(\HDone{n}) \geq \Omega(\log N)$? 
\end{problem}

\medskip
\noindent
\emph{Expander graphs.}  
As a second application, let $G$ be an $(N, d, \lambda)$-spectral expander: a $d$-regular graph with the second largest eigenvalue of the adjacency matrix not exceeding $\lambda$, denote its adjacency matrix with $M_G$. 
We prove:  
\begin{restatable}{theorem}{lbforexp} 
\label{thm:lower-bound-exp}
    $\spiky(M_G) \geq \Omega\!\left(\min\!\left(\tfrac{d}{\lambda}, \log \tfrac{N}{d^2}\right)\right).$
\end{restatable}
This implies (see \cref{cor:log-lb-for-expanders}) an explicit $\Omega(\log N)$ lower bound for the spiky rank of adjacency matrices of $(N, \Theta(\log^2 N), \Theta(\log N))$-spectral expanders. 

\medskip
\noindent
\emph{Inner Product.}  The Boolean inner product matrix is defined by $\IP_n(x,y) = \sum_{i \in [n]} x_i y_i \bmod 2$ for $x,y \in \{0,1\}^n$. For $\IP_n$ we obtain bounds that are weak but improve over the trivial ones:
 $$\frac{n}{\log n} \leq \spiky(\IP_n) \leq 2^{(1-\delta)n},$$
for some constant $\delta > 0$. The lower bound is established in \cref{sec:ip_bound}, while the upper bound follows from non-rigidity results (\cref{sec:intro-rigidity}).

\medskip
\noindent
\emph{Disjointness.} The Disjointness matrix $\Disj_n$ is the $2^n \times 2^n$ Boolean matrix whose rows and columns are indexed by $\{0,1\}^n$, with $\Disj_n(x,y)=1$ iff $x$ and $y$ have no common $1$'s. Similarly, for the Disjointness matrix we observe the same lower bound as for $\IP_n$ (\cref{sec:ip_bound}), so combined with Williams' upper bound \cite{williamsOrthogonal2024} we have
\[ \frac{n}{\log n} \le \spiky(\Disj_n) \le 5^{n/5}.\]

We believe the true lower bounds for $\Disj_n$ and $\IP_n$ are much closer to the upper bounds, and in particular ask:

\begin{problem}
Is it true that $\spiky(\IP_n) = \omega(n)$ or $\spiky(\Disj_n) = \omega(n)$? 
\end{problem}

\begin{figure}[h!]
    \centering
    \scalebox{0.9}{\definecolor{color2}{HTML}{B5E5CF}
\definecolor{color1}{HTML}{FCB5AC}
\definecolor{color3}{HTML}{888888}
\definecolor{color4}{HTML}{2222FF}
\begin{tikzpicture}[scale=0.3]

\def\scaleWidth{50}
\def\scaleHeight{1}
\def\fillFraction{0.3} 
\def\cornerRadius{0.7}

\def\cornerRadius{0.5} 


\begin{scope}
    \clip
      (\cornerRadius,0) -- (\fillFraction*\scaleWidth,0) -- (\fillFraction*\scaleWidth,\scaleHeight)
      -- (\cornerRadius,\scaleHeight) arc[start angle=90,end angle=180,radius=\cornerRadius cm] -- cycle;

    \foreach \i in {-0.3,0,...,50} { 
        \draw[very thick, color1!70] (\i,0) -- (\i-1,\scaleHeight);
    }
\end{scope}

\newcommand{\placeflag}[6]{%
  \pgfmathsetmacro{\x}{#1*\scaleWidth}
  \ifx#4a 
    \draw[very thick,#5] (\x,0) -- (\x,\scaleHeight+#6);
    \node[draw=black,fill=white,rotate=30,anchor=south west,inner sep=2pt,rounded corners=3pt] 
      at (\x,\scaleHeight+#6) {#2};
    \node[below] at (\x,0) {#3};
  \else 
    \draw[very thick,#5] (\x,\scaleHeight) -- (\x,-#6);
    \node[draw=black,fill=white,rotate=30,anchor=north west,inner sep=2pt,rounded corners=3pt] 
      at (\x,-#6) {#2};
    \node[above] at (\x,\scaleHeight) {#3};
  \fi
}

\placeflag{0.10}{$\HDone{n}$: \cref{thm:lower-bound-hd1}}{$\sqrt{\log N}$}{a}{color1}{5.5}
\placeflag{0.20}{$\IP_n, \Disj_n$: \cref{sec:ip_bound}}{$\tfrac{\log N}{\log \log N}$}{a}{color1}{5.5}
\placeflag{0.30}{Expanders: \cref{cor:log-lb-for-expanders}}{$\log N$}{a}{color1}{6}
\placeflag{0.40}{$\underset{\text{\cref{sec:circuits}}}{\Sigma \circ \text{ReLU lower bounds}}$}{$\log^{5/2} N$}{b}{color1}{13}
\placeflag{0.85}{$\underset{\text{\cref{sec:rigidity}}}{\text{Rigidity bounds}}$}{$\tfrac{N}{2^{(\log \log N)^{O(1)}}}$}{b}{color1}{8}
\placeflag{0.95}{$\underset{\text{\cref{thm:all_matrices_are_hard}}}{\text{Random Boolean Matrix}}$}{$\tfrac{N}{\log N}$}{a}{color1}{8}

\placeflag{0.60}{$\Disj_n$ \cite{williamsOrthogonal2024}}{$N^{0.47}$}{a}{color2}{8}
\placeflag{0.70}{$\IP_n$ \cite{AW17}}{$N^{1-\delta}$}{a}{color2}{8}

\draw[thick]
  (\cornerRadius,0) -- (\scaleWidth-\cornerRadius,0)
  arc[start angle=270,end angle=360,radius=\cornerRadius cm] --
  (\scaleWidth,\scaleHeight-\cornerRadius)
  arc[start angle=0,end angle=90,radius=\cornerRadius cm] --
  (\cornerRadius,\scaleHeight)
  arc[start angle=90,end angle=180,radius=\cornerRadius cm] --
  (0,\cornerRadius)
  arc[start angle=180,end angle=270,radius=\cornerRadius cm] -- cycle;
\end{tikzpicture}}
    \caption{
    \definecolor{color2}{HTML}{B5E5CF}
\definecolor{color1}{HTML}{FCB5AC}
\definecolor{color3}{HTML}{888888}
\definecolor{color4}{HTML}{2222FF}
    The picture illustrates the upper and lower bounds for spiky rank: above the scale there are known \colorbox{color1}{lower bounds} and \colorbox{color2}{upper bounds}. Below the scale are consequences of the corresponding lower bounds. The dashed part of the scale represents the state of the art of explicit lower bounds. Here $N$ is always the size of the matrix, so for $\HDone{n}$, $\Disj_n$, and $\IP_n$, $N = 2^n$. }
    \label{fig:hdone}
\end{figure}
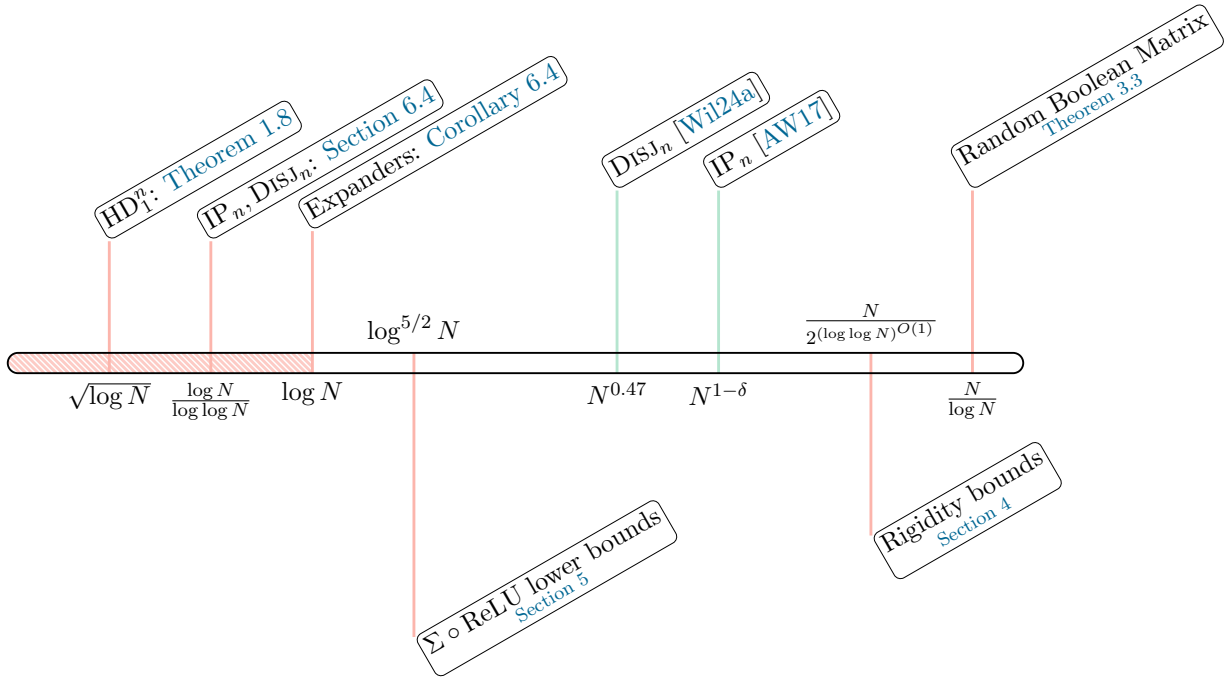

\subsubsection{Relations to other matrix parameters}\label{sec:matrix_parameters}
When introducing a new matrix parameter, a natural and useful question to understand is how it compares to existing ones: is it stronger, weaker, or qualitatively different? 

\paragraph{Blocky rank.} Since spiky rank can be seen as an algebraic strengthening of blocky rank, the first step is to ask whether this additional flexibility actually increases its power. Is blocky rank equivalent to spiky rank? 

\begin{question}
Does every Boolean matrix $M$ satisfy $\br(M) \leq \poly(\spiky(M))$?
\end{question}

We answer this question qualitatively: there is a dimension-free relation between spiky and blocky rank:

\begin{restatable}{theorem}{thmbrlefspr}
\label{thm:brlefspr}
    For any Boolean matrix $M$, 
    $\br(M)\leq \spiky(M)^{O(\spiky(M))}.$
\end{restatable} 

This situation is reminiscent of tantalizing \emph{log-rank conjecture}. There, too, a combinatorial parameter - the partition number - is vastly strengthened by adding algebraic property via matrix rank. Similar to  \cref{thm:brlefspr}, an exponential upper bound is known, and the challenge is whether this relation can be improved to a quasi-polynomial one.

\paragraph{$\gamma_2$-norm.} 

Since our main goal is to establish lower bounds on spiky rank, it is natural to compare it with matrix parameters that are more amenable to analysis -- particularly matrix norms. A well-studied example, with many applications in complexity theory, is the \emph{$\gamma_2$-norm} (see \cite{balla2025factorization} and references within). For a matrix $M$, it is defined as 
$$\gamma_2(M)=\min_{U, V : M=UV}\|U\|_{\row}\|V\|_{\col},$$
where $\|U\|_{\row}$ denotes the maximum $\ell_2$-norm of the rows of $U$, and $\|V\|_{\col}$ denotes the maximum $\ell_2$-norm of the columns of $V$.

Strong lower bounds are known for $\gamma_2$-norm for functions of interest; for instance, $\gamma_2(\IP_n) \ge N^{\Omega(1)}$, where $\IP_n(a,b) = \sum_{i\in[n]} a_i b_i \bmod 2$ (see e.g. \cite{lee2009lower}). This motivates comparing $\gamma_2$ with rank-based parameters such as blocky rank and spiky rank. For blocky rank, it is known that for every Boolean matrix $M$, $\br(M) \geq \log \gamma_2(M)$ \cite[Proposition 3.1]{HHH23}. For spiky rank, while an analogous inequality currently seems out of reach, we observe a simple non-explicit separation (\cref{sec:gamma-2-vs-br}):
\begin{theorem} \label{thm:gamma-2-vs-br}
    There exists a Boolean matrix $M$ such that $\spiky(M) \ge \Omega(\gamma_2^2(M))$.
\end{theorem}

The reverse direction -- whether blocky rank is upper-bounded by a function of $\gamma_2$-norm -- is a conjecture asked in $\cite{HHH23}$, and is equivalent to a conjecture in operator theory. Very recently, Goh and Hatami \cite{goh2025block} proved that for $N \times N$ Boolean matrix $M$, $\br(M) \leq 2^{\gamma_2(M)} (\log N)^2$, and the same bound also applies to spiky rank.

\paragraph{Sparsity.} One of the most basic and easily computable parameters of a matrix is its \emph{sparsity}, the number of non-zero entries. We show that sparse matrices necessarily have low spiky rank, a fact that becomes useful when relating spiky rank to matrix rigidity (\cref{lem:rigidity_vs_sparsity}). In particular, for any real matrix $M$ with  sparsity $m$, we prove $$\spiky(M) \leq \sqrt{m}. $$ For Boolean matrices, this bound can be sharpened, and even expressed directly in terms of blocky rank (\cref{thm:sparsity_vs_br_Boolean}). Specifically, for any Boolean matrix $M$ with sparsity $m$, it holds 
$$\br(M)\leq O\left(\frac{\sqrt{m}\log\log m}{\log m}\right).$$

\paragraph{Sign and approximate measures.} 

Some computational models are naturally characterized not by exact matrix parameters but by their approximate or sign variants. For instance, randomized communication complexity is lower-bounded by approximate rank, while unbounded-error communication complexity is captured by sign-rank.
Formally, for a matrix measure $\mu\colon \{0,1\}^{N \times N} \to \mathbb{R}$, we define its approximate version by
\[ \mu_\epsilon(M) \coloneqq \inf\{\mu(M') \mid \|M-M'\|_\infty \le \epsilon \text{ , } M' \in \R^{N \times N} \}.\]
Similarly, the sign version of $\mu$ is $\mu_\pm(M) \coloneqq \inf\{\mu(M') \mid \sign(M) = \sign(M')\}$. 

Approximate measures tend to be stronger than their exact counterparts, and often defy intuition developed for them. For example, the approximate version of the log-rank conjecture is false \cite{CMS20}. In case of spiky rank,  it is sharply separated from all related approximate and sign measures of a matrix: $\spiky(\HDone{n}) \geq \sqrt{\log N}$  (\cref{thm:lower-bound-hd1}), while $\rk_{\pm}(\HDone{n}) = O(1)$ \cite{GHIS25}, $\gamma_{2,\varepsilon}(\HDone{n}) = O(1)$ \cite{CHZZ24}, $\br_\epsilon(\HDone{n}) = \poly(\log \log N)$ (see \cref{sec:approx-br-bound}). 

This highlights an interesting contrast between the spiky and blocky rank. For blocky rank, there is a dimension-free relation with its sign version: we show in \cref{sec:br-vs-sign-br} that $\br(M) \le 2^{\br_\pm(M)}$. By contrast, for spiky rank we have $\spiky(\HDone{n}) \ge \Omega(\sqrt{\log N})$ but $\spiky_\pm(\HDone{n}) \le \rk_\pm(\HDone{n}) = O(1)$, ruling out any dimension-free relation between $\spiky$ and $\spiky_\pm$.

Thus, proving lower bounds for sign spiky rank is an even harder challenge. At the same time, this opens a promising direction to seek explicit versions of the known non-explicit bounds below the logarithmic barrier of our framework. For example, we know that the adjacency matrix of a random degree-$d$ graph has sign spiky rank $\Omega(d)$ (\cref{thm:large-br-bounded-deg}). A natural candidate for derandomization is a spectral expander, leading to the following question:
\begin{question}
    What is $\spiky_\pm(M_G)$ for a spectral expander $G$? 
\end{question}

\paragraph{Organization.} \cref{sec:prelim} introduces preliminaries. In \cref{sec:nonexplicit bounds}, we prove the existence of Boolean and real matrices with high spiky rank. The applications to matrix rigidity and circuit complexity are presented in \cref{sec:rigidity} and \cref{sec:circuits}, respectively. \cref{sec:explicit bounds} establishes spiky rank bounds for specific explicit matrices. In \cref{sec:relations}, we explore connections between spiky rank and other matrix parameters. Finally, \cref{sec:approx_sign} demonstrates separations between the approximate and sign versions of blocky and spiky rank.

\section{Preliminaries} \label{sec:prelim}

Our convention in this paper is to use capital letters for matrix dimensions, e.g.\ $A \in \mathbb{R}^{M \times N}$. We use the notation $\circ$ for the entrywise product. The notation $[k]$ represents the set of integers $\{1, \ldots, k\}$.

For a field $\F$, we denote $0_\F$ for the additive identity and $1_\F$ for the multiplicative identity. Addition, multiplication, and the entrywise product of matrices over $\F$ are denoted $+_\F$, $\cdot_\F$, and $\circ_\F$ respectively.

For a matrix $M \in \mathbb{F}^{N \times N}$ its sparsity $\spar(M)$ is defined as the number of its non-$0_\F$ elements.

We write $\mathbb{1}_S$ for $S \in [N] \times [N]$ for the matrix in $\{0,1\}^{N \times N}$, where $\mathbb{1}_S[i,j] = 1$ iff $(i,j) \in S$.

\paragraph{Blocky and spiky measures.} As some of our results generalize to arbitrary fields, we rewrite \cref{def:matrices,def:ranks} here with respect to an arbitrary field $\F$.

\begin{definition}
    A matrix $M \in \F^{M \times N}$ is \emph{blocky} if there is an integer $k$ and two collections of disjoint sets $S_1, \dots, S_k \subseteq [M]$ and $T_1, \dots, T_k \subseteq [N]$ such that $M[i, j] = 1_\F$ if $(i,j) \in \bigcup_{\ell \in [k]} S_i \times T_i$ and $M[i, j] = 0_\F$ otherwise.

    A matrix $M \in \F^{M \times N}$ is \emph{spiky} if there exists a blocky matrix $B$ and a rank-one matrix $C$ (where rank is taken over $\F$) such that $M = B \circ_\F C$.
\end{definition}
\begin{definition}
    The \emph{blocky rank} of a matrix $M \in \F^{M \times N}$, denoted $\br^{\F}(M)$, is the minimum $r$ such that $M = \sum_{i \in [r]} \alpha_i B_i$ where each $\alpha_i \in \F$ and each $B_i$ is blocky.
    
    The \emph{spiky rank} of a matrix $M \in \F^{M \times N}$, denoted $\spiky^{\F}(M)$, is the minimum $r$ such that $M = \sum_{i \in [r]}S_i$ where each $S_i$ is spiky.
\end{definition}

When the underlying field is $\R$, we write simply $\br(M)$ and $\spiky(M)$.

Here are some easy facts about spiky rank:

\begin{fact}[Subadditivity]
    For any field $\F$ and matrices $A, B \in \F^{M \times N}$, $\spiky^{\F}(A + B) \leq \spiky^{\F}(A) + \spiky^{\F}(B)$.
\end{fact}

\begin{fact}[Upper bounds]
    For any field $\F$ and matrix $A \in \F^{M \times N}$, $\spiky^{\F}(A) \leq \rank^\F(A)$ (where $\rank^\F$ is the rank over $\F$) and $\spiky^{\F}(A) \leq \br^{\F}(A)$.
\end{fact}

\paragraph{Sign blocky rank and sign spiky rank.} We study the sign versions blocky rank and spiky rank over the reals. In the following definition, we assume that the matrices have no 0-entries. For a matrix $A$, the value $\sign(A)$ is the $\sign$ function applied element-wise: 

\[ \sign(A)[i, j] = \begin{cases} 1 & \mbox{if } A[i, j] > 0\\ -1 & \mbox{if } A[i, j] < 0 \end{cases}. \]

\begin{definition} \label{def: sign variants}
    Let $A$ be a matrix in $\mathbb{R} ^{M \times N}$.
    \begin{itemize}
        \item The \emph{sign blocky rank} of $A$, denoted $\sbr(A)$, is defined as
    \[\sbr(A) = \min_{B \in \R^{M \times N}}\{\br(B): \sign(A)=\sign(B)\}.\]

        \item     The \emph{sign spiky rank} of $A$, denoted $\spiky_{\pm}(A)$, is defined as
    \[\spiky_{\pm}(A) = \min_{B \in \R^{M \times N}}\{\spiky(B): \sign(A)=\sign(B)\}.\]
    \end{itemize}    
\end{definition}

\begin{fact}[Upper bounds for sign measures.]
    For any matrix $A \in \R^{M \times N}$, $\sbr(A) \leq \br(A)$ and $\spiky_{\pm}(A) \leq \spiky(A)$.
\end{fact}

\section{Non-explicit lower bounds for spiky rank} \label{sec:nonexplicit bounds} 
\subsection{Lower bound for random Boolean matrices}

 We first count the number of $N \times N$ Boolean matrices that have spiky rank $r$ for a fixed number $r \leq N$. It is sufficient to count the number of  $\pm 1$-matrices of with spiky rank $r$. Indeed, if $M \in \{0,1\}^{N \times N}$,  then  its signed version $M' \in \{\pm 1\}^{N \times N}$ defined as $M' = J -2 M$ with $J$ being the $N \times N$ all-ones matrix, has $\spiky(M') \leq r+1$.

In fact, we will actually count the number of $\pm1$-matrices with sign spiky rank $r$ (\cref{def: sign variants}). As this measure is upper-bounded by spiky rank, this gives us an estimate for the number of matrices with spiky rank $r$ as well. 

The idea is to use Warren's theorem \cite{warren1968lower} to count the number of possible sign patterns. The same idea and argument were used to count the number of matrices with sign rank $r$ in \cite{alon2016sign}). Before stating the theorem, let us set it up first.

Let $P_1, P_2, \ldots, P_m$ be real polynomials in $\ell$ variables each. For $x \in \R^{\ell}$, such that no polynomial among $\{P_i\}$ vanishes at $x$, its sign pattern is the vector 
$$\big(\sign(P_1(x)), \sign(P_2(x)), \ldots, \sign(P_m(x)) \big) \in \{-1,1\}^m.$$

\begin{theorem}[\cite{warren1968lower}]\label{thm:warren}
Let $P_1, P_2, \ldots, P_m$ real polynomials, each in $\ell$ variables and of degree at most $k$. If $m \geq \ell$, then the number of different sign pattern for $x$ ranging over all $\R^{\ell}$ is at most $\left(4ekm/\ell\right)^{\ell}$.
\end{theorem}
\begin{lemma}\label{lem:counting_spr}
   The number of $N \times N$ sign matrices with sign spiky rank $r \le N/2$ is at most $N^{6rN}$. 
\end{lemma}
\begin{proof}
Let $M$ be a $N\times N$ sign matrix with sign spiky rank $r$.
Let $A \coloneqq A_1 + \dots + A_r \in\R^{N\times N}$ be the sign representation of $M$: $M=\sign(A_1 + \dots + A_r)$. Let us represent each $A_i = B_i \circ C_i$ where $B_i$ is a blocky matrix and $C_i$ is a rank-$1$ matrix. There are at most $N^{4rN}$ collections $B_1, \dots, B_r$, so let us fix one and count the number of potential matrices $M$ that appear as we change $C_1, \dots, C_r$. 

With fixed $B_1, \dots, B_r$ each entry of $A$ can be written as a degree-2 polynomial in $2 N r$. Let for each $i \in [r]$, $C_i = x_i^Ty_i$ for some vectors $x_i,y_i \in \R^{N}$. Then for $j,k \in [N]$ we have $A_{jk} = \sum_{i \in [r]} (B_i)_{jk} \cdot x_{ij} y_{ik}$. Then by
 \cref{thm:warren} we get that the number of different sign matrices $M$ of form $\sign(B_1 \circ C_1 + \dots + B_r \circ C_r)$ with fixed $B_1, \dots, B_r$ is at most $\left(4e \cdot \tfrac{2N^2}{2rN}\right)^{2rN} \leq N^{2rN}$, so in total there are at most $N^{4rN} \cdot n^{2rN}$ matrices of sign spiky rank $r$.
\end{proof}

Combining the lemma with the observation that $\spiky_{\pm}(M) \leq \spiky(M)$, we get that a uniformly random Boolean matrix has high spiky rank with high probability. More precisely,
\begin{theorem}\label{thm:all_matrices_are_hard}
    If $M$ is a $N \times N$ uniformly random Boolean matrix, then 
    \[\Pr \Big[\spiky(M) \geq \frac{N}{12\log N}\Big] \geq \Pr \Big[\spiky_{\pm}(M) \geq \frac{N}{12\log N}\Big] \geq 1-{2^{-\frac{N^2}{2}}}. \] 
\end{theorem}
\begin{proof}
    The number of matrices of sign spiky rank less than $N/(12\log N)$ is at most $N^{6N^2/(12\log N)} \le 2^{N^2/2}$ by \cref{lem:counting_spr} and the total number of $N \times N$ sign matrices is $2^{N^2}$.
\end{proof}

\subsection{Lower bound for real matrices}
The spiky rank of any $N \times N$ real matrix is trivially bounded above by $N$, since the real rank of a matrix is at most $N$, and every rank-one matrix is also a spiky matrix.
For Boolean matrices, even blocky rank can be bounded above by $O({N}/{\log N})$. However, for general real matrices, we show that the trivial upper bound is essentially tight: there exist matrices whose spiky rank is $\Omega(N)$. The key point in the proof as in \cref{lem:counting_spr} is that all spiky matrices with a fixed blocky pattern can be seen as evaluations of a polynomial.

We will need the following standard probabilistic notions: a distribution $\mu$ over $\R^m$ is \emph{absolutely continuous} if there exists an integrable density function $f\colon \R^m \to \R^+$ such that for every Lebesgue-measurable set $A \subseteq \R^m$ we have $\mu(A) = \int_A f d\lambda$ where $\lambda$ is the Lebesgue measure. Examples of absolutely continuous distributions include the uniform distribution over $[0,1]^m$ and the Gaussian distribution. 

\begin{theorem}\label{thm:linear-spiky-rank}
    Let $\mu$ be any absolutely continuous probability distribution over $\mathbb{R}_{N \times N}$. Then \[\Pr_{M \sim \mu}[\spiky(M) \ge N/2] = 1.\] 
    In particular this holds for $M \sim [0,1]^{N \times N}$.
\end{theorem}
\begin{proof}
    It suffices to show that the Lebesgue measure of the set of all matrices with spiky rank at most $N/2-1$ is zero. This immediately implies the existence of a matrix with spiky rank at least $N/2$.

    Let $M \in \mathbb{R}^{N \times N}$ be a matrix with $\spiky(M) = r$. Then $M = \sum_{i \in [r]} S_i$, where each $S_i$ is a spiky matrix. Each spiky matrix $S_i$ can be viewed as an entrywise product of a blocky matrix $B_i$ and a rank-1 matrix $a_i^T b_i$ for some vectors $a_i,b_i \in \mathbb{R}^n$. 
    Fix a collection of blocky matrices $B_1, \dots, B_r \in \{0,1\}^{N \times N}$, and define the corresponding set of spiky rank-$r$ matrices with the blocky pattern $B_1,\dots,B_r$:
    
    $$\mathcal{S}_{B_1, \dots, B_r} =
    \Big\{\sum_{i \in [r]} (a_i^T b_i) \circ B_i \mid a_1,\dots,a_r,b_1,\dots,b_r \in \mathbb{R}^N\Big\},$$ where $\circ$ is the Hadamard (entrywise) product. 
    Then, the set of all matrices in $\mathbb{R}^{N \times N}$ of spiky rank at most $r$ is equal to the union over all such choices of $B_1, \dots, B_r$:
    $\bigcup_{B_1, \ldots B_r} \mathcal{S}_{B_1, \dots, B_r}.$ Since the number of blocky matrices is finite (trivially it does not exceed $2^{N^2}$), to show that the union set has measure zero it is sufficient to show that each set $\mathcal{S}_{B_1,\dots,B_r}$ has measure zero whenever $r\le N/2-1$. 

    To show that $\mathcal{S}_{B_1, \dots, B_r}$ has Lebesgue measure zero, define a smooth function $f\colon \R^{2r \cdot n} \to \R^{n \times n}$ as follows: for  $a_1, \ldots a_r, b_1, \ldots, b_r \in \R^{N}$ and $i,j \in [N]$, let
    \[f(a_1, b_1, \dots, a_r, b_r)[i,j] = \sum_{\ell\in[r]} 1_{B_{\ell}[i,j] = 1} \cdot a_{\ell}[i] \cdot b_{\ell}[j].\]
    The image of $f$ is exactly $\mathcal{S}_{B_1, \dots, B_r}$.
    \begin{fact}
    \label{fact:measure-zero-mapping}
        Suppose $f_1, \dots, f_n\colon \R^m \to \R$ are polynomials and $A \subseteq \R^m$ is a set with Lebesgue measure zero. Then $f(A) \coloneqq \{(f_1(x),\dots,f_n(x)) \mid x \in A\}$ also has Lebesgue measure zero.
    \end{fact}
    \begin{proof}
        Since polynomials are $C^1$-differentiable (in fact, infinitely differentiable) we apply the more general standard result (see e.g. \cite[Proposition~1.3]{GG73}) that states that any $C^1$-differentiable transformation maps a measure-zero set to a measure-zero set.
    \end{proof}
    
    Since $\R^{2r \cdot N}$ has measure zero when $2rN < N^2$ and $\mathcal{S}_{B_1,\dots,B_r} = f(\R^{2r \cdot N})$, by \cref{fact:measure-zero-mapping} $\mathcal{S}_{B_1, \dots, B_r}$ has measure zero for all $r < N/2$, and so does their finite union. 
\end{proof}

\section{Connection to matrix rigidity} \label{sec:rigidity}

\begin{definition}[Rigidity]
For a matrix $M \in \F^{N \times N}$ over a field $\F$, the \emph{rank-$r$ rigidity} $\cR^\F_M(r)$ is the minimum number of entries that must be modified to reduce the rank of $M$ to $r$. Formally,
\[
\cR^\F_M(r) = \min \{\, s \mid M = A +_\F C, \; \rank^\F(A) \leq r, \; \spar(C) \leq s \,\}.
\]
\end{definition}

We show that matrices with large spiky rank must also be highly rigid. Our main result is the following generalization of \cref{thm:spr_lower_bound_on_rigidity} to arbitrary fields.

\begin{theorem} \label{thm:spr-rigidity-fields}
    Let $M$ be a matrix over a field $\F$ and $0 < r \leq \spiky^{\F}(M)$. Then, 
    $$\cR^{\F}_M(r) \geq \frac{(\spiky^{\F}(M) - r)^2}{4}.$$
\end{theorem}

The key ingredient is that sparse matrices have small spiky rank.
\begin{lemma}\label{lem:rigidity_vs_sparsity}
Let $M \in \F^{N \times N}$ be a matrix over an arbitrary field $\F$ with at most $m$ non-zero entries. Then,
$$\spiky^{\F}(M) \leq 2\sqrt{m}.$$ 
\end{lemma}
We defer the proof to \cref{sec:sparsity}. This immediately yields the following relationship between spiky rank and matrix rigidity.

\begin{proof}[Proof of \cref{thm:spr-rigidity-fields}]
    Assume $\cR^{\F}_M(r) = s$ for some $r \leq N$. Then $M$ can be written as $M = A +_{\F} C$, where $A \in \F^{N \times N}$ has  $\rank^{\F}(A)=r$ and $C \in \F^{N \times N}$ has at most $s$ non-zero entries. 
    By the subadditivity of spiky rank over $\F$, we have
    $$ \spiky^{\F}(M) = \spiky^{\F}(A+_{\F}C) \leq \spiky^{\F}(A) + \spiky^{\F}(C) \leq r + 2\sqrt{s},$$
    where the last inequality follows from \cref{lem:rigidity_vs_sparsity} and the trivial observation that $\spiky^{\F}(M) \leq \rank^{\F}(M)$ for all $M$.
    Rearranging gives $s \geq (\spiky(M) - r)^2/4$, which completes the proof.
\end{proof}

\begin{remark}
    The condition $0 <r \leq \spiky^{\F}(M)$ in~\cref{thm:spr-rigidity-fields} limits its applicability for proving intermediate rigidity bounds. For instance, for Valiant’s target rank $\epsilon N$, the theorem applies only to matrices already known to have spiky rank $\Omega(N)$. 
\end{remark}

\begin{remark}
\cref{lem:rigidity_vs_sparsity}, thus also \cref{thm:spr-rigidity-fields}, is tight for real matrices: a random real matrix has spiky rank $\Omega(N)$ by \cref{thm:linear-spiky-rank}, and thus sparsity $\Omega(N^2)$ and rigidity $\Omega(N^2)$, the maximum possible. For Boolean matrices, we can slightly improve the sparsity and spiky rank tradeoff (see \cref{thm:sparsity_vs_br_Boolean}).
\end{remark}

To discuss the implications of \cref{thm:spr-rigidity-fields} for communication complexity, 
we use the following result of Razborov \cite{razborovRigidity} and Wunderlich \cite{wunderlich2012theorem}.

\begin{theorem}[{\cite[Theorem 5.15]{wunderlich2012theorem}}]\label{thm:wunderlich}
    Let $\F$ be a finite field and let $M$ be a $N \times N$ matrix. If \[\cR_M^\F\left(2^{(\log\log N)^{\omega(1)}}\right) \geq \Omega\left(\frac{N^2}{2^{(\log\log N)^{O(1)}}}\right).\]

Then $M \notin \clsPH^{cc}$.
\end{theorem}

    As an immediate consequence we obtain:

    \begin{corollary} \label{cor:spiky_ph}       
        Suppose a ${N \times N}$ Boolean matrix $M$  satisfies 
        $$\spiky^{\F}(M) \geq \frac{N}{ 2^{(\log\log N)^{O(1)}}},$$ where $\F$ is the finite field. Then $M \notin \clsPH^{\text{cc}}$. 
    \end{corollary}
    \begin{proof}
    Applying \cref{thm:spr-rigidity-fields} with $r = \spiky^\F(M)/2$, we obtain
    \[
        \cR_M^\F\!\left(\tfrac{N}{2^{(\log\log N)^{O(1)}}}\right) 
        \;\;\geq\;\; \Omega\!\left(\frac{N^2}{2^{(\log\log N)^{O(1)}}}\right).
    \]
    Since 
    \[
        \cR_M^\F\!\left(2^{(\log\log N)^{\omega(1)}}\right) 
        \;\;\geq\;\; \cR_M^\F\!\left(\tfrac{N}{2^{(\log\log N)^{O(1)}}}\right),
    \]
    the claim follows directly from \cref{thm:wunderlich}.
\end{proof}

\section{Connection to circuit complexity} \label{sec:circuits}
\begin{definition}
A \emph{$\relu$ gate} (rectified linear unit) is a function $\ell\colon \{0,1\}^n \to \R^{+}$ such that there exist a vector $w \in \R^n$ and a scalar $a \in \R$ such that for all $x \in \{0,1\}^n$
$$\ell(x) = \max \{0, \langle w,x \rangle + a\}.$$
\end{definition}

Intuitively, a ReLU function is a \textit{weighted} linear threshold function. For our purposes it will be convenient  to view a ReLU gate as a matrix:
a $\relu$ gate is a matrix $L\colon \{0,1\}^n \times \{0,1\}^n \to \R^{+}$ such that there exist vectors $w_1, w_2 \in \R^n$ and a scalar $\alpha \in \R$ for which
$$L(x,y) = \max \{0, \langle w_1,x \rangle + \langle w_2,y \rangle + \alpha\}.$$

We now show that $\Sigma \circ \relu$ \-circuit complexity -- the number of $\relu$ gates -- of any function is upper bounded by its spiky rank up to an $O(n)$ factor. As a first step, we bound the spiky rank of a single ReLU gate.

\begin{claim}
    Let $L\colon \{0,1\}^n \times \{0,1\}^n\to \R^{+}$ be a $\relu$ gate. Then $$\spiky(L) \leq 3n + 3.$$
\end{claim}
\begin{proof}
Let $w_1, w_2 \in \R^n$ and $\alpha \in \R$ be such that $L(x,y) = \max \{0, \langle w_1,x \rangle + \langle w_2,y \rangle + \alpha\}.$
Denote the underlying function $L'(x,y) = \langle w_1,x \rangle + \langle w_2,y \rangle + \alpha $. Observe that $L'$ is a sum of three rank-$1$ matrices $\langle w_1,x \rangle$, $\langle w_2,y \rangle$ and $\alpha$ (each depending only on $x$, $y$ or being constant), so $\rank(L') \leq 3$.

Define $T_{L}$ as the threshold function corresponding to $L$ as follows: 
$$
T_{L}(x,y) =
\begin{cases}
    1, \text{ if } L'(x,y) > 0\\ 
    0, \text{ otherwise.}
\end{cases}
$$  
 Avraham and Yehudayoff~\cite{DY24} show that any threshold function on $\{0,1\}^n \times \{0,1\}^n$ can be written as a disjoint union of at most $n+1$ blocky matrices:
$$T_{L}=B_1 + \ldots + B_{n+1},$$
where $B_1, \ldots, B_{n+1}$ are disjoint blocky matrices, i.e., their supports are disjoint rectangles in the matrix.  
We now define weighted matrices $R_1, \ldots, R_{n+1}$ corresponding to $B_i$'s by placing the original values of $L$ on the support of each $B_i$:
$$
R_i(x,y) =
\begin{cases}
L(x,y), \text{ if } B_i(x,y)=1\\
0, \text{ otherwise.}
\end{cases}
$$
 Now since $B_i$'s are disjoint, we have $L = R_1 +\ldots + R_{n+1}.$
The only problem now is that $R_i$ might not be a spiky matrix as it might have blocks of rank greater than 1. However, the blocks of $R_i$ are submatrices of $L'$, so each block of $R_i$ has rank at most $3$. If there is a block $R_i$ of rank greater than $1$, then we can decompose $R_i$ into at most three matrices that have rank-$1$ on that block.  If multiple blocks in the same $R_i$ have rank greater than 1, then we can still partition $R_i$ into at most 3 matrices that have rank-1 on all blocks, because the blocks are on disjoint rows and columns. This shows that each $R_i$ can be decomposed into at most three spiky matrices. Hence, $\spiky(M) \leq 3(n+1)$.
\end{proof}
\begin{proposition}\label{thm:spr_vs_relu}
For $M\colon \{0,1\}^n \times \{0,1\}^n \to \R$, let $s$ be the size of a $\Sigma \circ \relu$ \-circuit that computes $M$. Then 
$$\frac{\spiky(M)}{3(n+1)} \leq s.$$
\end{proposition}
\begin{proof}
   For the lower bound, let $M = \sum_{i=1}^s L_i$, where $L_i$ are ReLU functions. Then,
\[\spiky(M) \leq \sum_{i=1}^s \spiky(L_{i}) \leq s \cdot (3n+3).\qedhere\]
\end{proof}

\section{Explicit lower bounds for spiky rank} \label{sec:explicit bounds}
\subsection{Lower bound framework}
In this section we give a framework for proving spiky ranks lower bounds that applies to thin matrices such as $\HDone{n}$ and adjacency matrices of expander graphs.

\begin{theorem}\label{thm:sprlbframe}
    Let $M \in \R^{N\times N}$, $s, k, D$ be positive integers and $\gamma \in (0, 1)$. If 
    \begin{enumerate}[noitemsep,label=(P\arabic*)]
        \item For every $S \subseteq [N]$ and $T\subseteq [N]$ with $|S|\leq s$ and $|T|\leq s$ we have $\spar(M|_{S\times T}) \le k \cdot (|S| + |T|)$. \label{item:thin}
        \item For every $M' \le M$ (entrywise) such that $\spar(M - M') \le \gamma \cdot \spar(M)$, there exist sets $S \subseteq [N]$ and $T\subseteq[N]$ such that $|S| = |T| = D$ and $M'|_{S\times T}$ is a permutation matrix. \label{item:permutation} 
    \end{enumerate}
    then $$\spiky(M) \geq \Omega\left(\min\left(\frac{\gamma \cdot \spar(M)}{ k N}, \log\left( \frac{Ds}{2N} \right)\right)\right).$$
\end{theorem}

We need the following lemma for the proof.
The proof of it is identical to the proof of  \cite[Lemma~12]{DY24}, but we include it here for completeness. 
\begin{lemma}\label{lem:non-intersecting-subset}
    Let $I_N$ be identity matrix of size $N\times N$. If there is a blocky (or spiky) matrix $B$, such that $B \land I_N = 0$ (i.e. there is no entry where both $I_N$ and $B$ are nonzero). Then there exists $V \subset [N]$, s.t. $|V| \geq N/4$ and $B|_{V\times V} = 0$.
\end{lemma}
\begin{proof}
    Let blocks of $B$ be $A_1 \times B_1, A_2\times B_2, \ldots, A_k \times B_k$. Define two random subsets of $[N]$ as follows. Let $\xi_1, \xi_2, \ldots, \xi_k$ be i.i.d. uniformly distributed on $\{0, 1\}$. 
    
    Let $S$ be the complement of $\bigcup_{a: \xi_a = 0} S_a$, and $T$ be the complement of $\bigcup_{a: \xi_a = 1} T_a$. Let $I = \{i \in [N]\mid (i, i) \in S\times T\}$. Note that $B|_{S\times T}$ is zero matrix. For each $i \in [N]$, the probability that $(i, i) \in S \times T$ is $1/4$, because $B_{ii} = 0$. Thus there is a choice for $S \times T$ so that $|I|\geq N/4$. For such $S$ and $T$, we define $V = S \cap T$.
\end{proof}

\begin{proof}[Proof of \cref{thm:sprlbframe}]
    \begin{figure}[h!]
    \centering
    \input{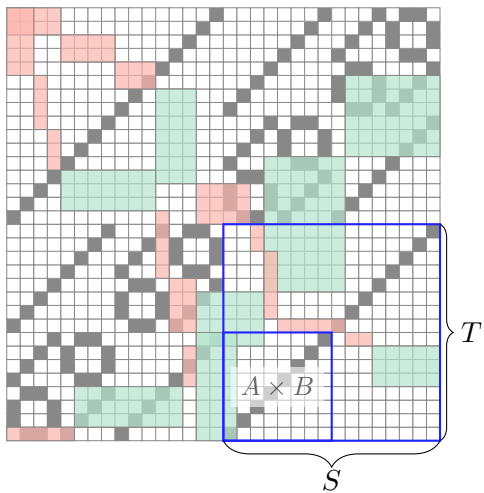}
    \caption{The picture illustrates the proof for the matrix $\HDone{5}$. The parts $T_i$ are colored red, $H_i$ are colored green. By \ref{item:thin} not too many $1$-cells (black ones) are covered by the small blocks, so by \ref{item:permutation} we find a permutation submatrix $S \times T$ with very few covered $1$-cells, and then further shrink it using \cref{lem:non-intersecting-subset} to $A \times B$ that does not contain any red cells at all. Then the rank of $\HDone{5}|_{A \times B}$ is maximal, but the total rank of large blocks is small, which is a contradiction.}
    \label{fig:hdone}
\end{figure}
    
    Let $M = \sum_{i\in[t]} S_i$. We call a block small, if the sum of its dimensions is at most $s$, and big otherwise. Assume the contrary and let $t \leq \frac{\gamma \cdot \spar(M)}{2 kN} $ and $t \leq \frac{1}{4}\log_2\left(\frac{Ds}{2N}\right)$.

    We represent $S_i = T_i + H_i$ where $T_i$ is comprised of the small blocks of the matrix $T_i$ and $H_i$ of the big ones (i.e. the sets of non-zero columns of $T_i$ and $H_i$ are disjoint, same for the non-zero rows).

    By \ref{item:thin}, each $T_i$ covers (blocky matrix \emph{covers} matrix entry, if this entry is non-zero in this matrix) at most $2kN$ $1$-entries of $M$, because the sum of block sizes in a spiky matrix does not exceed $2N$. This implies that all small matrices cover at most $2N tk\leq \gamma \cdot \spar(M)$ ones.

    By \ref{item:permutation} we get a permutation submatrix $S\times T$ and applying \cref{lem:non-intersecting-subset} to it $t$ times, we obtain a permutation
    submatrix $A \times B$ of $M$ with $|A| = |B| = D/2^{2t}$, s.t. for each $i$ the intersection of $T_i$ and $A\times B$ is empty. Thus, we obtain

    \begin{equation}
        \label{eq:rk-lowerbound}
    \rk \Big(M - \sum_{i \in [t]} T_i\Big)\geq \frac{D}{2^{2t}}
    \end{equation}

    On the other hand, each $H_i$ contains at most $2N/{s}$ blocks, that means 
    \begin{equation}
    \label{eq:rk-upperbound}
    \rk\Big(\sum_{i \in [t]} H_i\Big)\leq \sum_{i \in [t]} \rk(H_i) \le t \cdot\frac{2N}{s}.
    \end{equation}
    
    Combining \eqref{eq:rk-lowerbound} and \eqref{eq:rk-upperbound}, we get
    $D \cdot 2^{-2t}\leq 2Nt /{s},$
    which implies
    $\tfrac{Ds}{2N}\leq 2^{3t},$
    which contradicts our assumption.
\end{proof}

\subsection{Hamming-Distance-$1$}\label{sec:hd_bound}

$\HDone{n}$ is the function which determines whether the Hamming distance between the two input strings is 1. Its matrix is the adjacency matrix of the $n$-dimensional Hamming cube. 

\[ \HDone{n}(x, y) = \begin{cases} 1 & \mbox{if } \|x \oplus y\|_1 = 1 \\ 0 & \mbox{otherwise} \end{cases} \]

\label{sec:framework-applications}
\lbforhd*
We will make use of the following lemma.
\begin{lemma}[Lemma~3.10 in \cite{KSP20}]
\label{lem:haussler}
    If $G = (V,E)$ is a subgraph of a Hamming cube, then \[|V| \ge 2^{|E| / |V|}.\]
\end{lemma}
\begin{proof}[Proof of \cref{thm:lower-bound-hd1}]
    We use \cref{thm:sprlbframe} with $N = 2^n$, $M = \HDone{n}$, $s = 2^{\sqrt{n}}$, $k = 2\sqrt{N}$, $\gamma = 0.5$, $D = N/4$. We now prove two conditions of the theorem. \ref{item:thin} directly from \cref{lem:haussler} applied to $S \cup T$. 
    
    We proceed to prove \ref{item:permutation}. The $1$-entries of $\HDone{n}$ are partitioned into the following $2n$ subsets:
    \[K_{2i-j}=\{(x, y) \mid \HDone{n}(x,y)=1, x_i = j, y_i = 1-j\} \text{ for $i \in [n];$ $j \in \{0,1\}$}.\]
    Thus, each $1$-entry in $\HDone{n}$ belongs to exactly one of these subsets, and each subset has size $2^{n-1}$. The subsets have the following property: if $(x, y)$ and $(v, w)$ are distinct elements of $K_t$, then $\|x- w\|_1 > 1$ and $\|v-y\|_1 > 1$, because we have $x_i = v_i \neq y_i = w_i$ for some $i$. Thus, we have at least half of the $1$-entries present in $M'$, which means one of the groups has at least $\frac{1}{2}2^{n-1}$ elements in the intersection with $M'$. Without loss of generality, we assume that it is the subset $K_1$. Then we set $$S \subseteq \{1\}\times\{0,1\}^{n-1},~~~T\subseteq\{0\}\times\{0,1\}^{n-1},$$
    such that $S \times T$ covers exactly $2^{n-2}$ $1$-entries of $K_1$, which are present in $M'$, and they form a permutation submatrix.
    So, the desired $\Omega(\sqrt{n})$ bound is immediately implied by the theorem.
\end{proof}

\subsection{Expanders}\label{sec:expanders_bound}
 For a graph \(G\), we denote by \(M_G \in \{0,1\}^{N\times N}\) its adjacency matrix, i.e.,
\[
(M_G)_{u,v} \;=\;
\begin{cases}
1 & \text{if } (u,v)\in E,\\
0 & \text{otherwise}.
\end{cases}
\]
 We recall that an \((N,d,\lambda)\)-spectral expander is a \(d\)-regular graph \(G=(V,E)\) on \(N\) vertices such that the second largest eigenvalue of its adjacency matrix \(M_G\) is at most  \(\lambda\) in absolute value (see, for example, \cite{hoory06}).

\lbforexp*
\begin{proof}
    We use \cref{thm:sprlbframe} with $M = M_G$, $s = \tfrac{N\lambda}{10d}$, $k = \lambda$, $D = \tfrac{N}{4d}$, $\gamma = 1/2$. We check the theorem conditions.
    \ref{item:thin} holds by the expander mixing lemma 
        \[e(S, T) \leq \frac{d|S||T|}{N} + \lambda \sqrt{|S||T|} \leq \frac{ds}{N}\sqrt{|S||T|} + \lambda\sqrt{|S||T|} \leq 2\lambda\sqrt{|S||T|}\leq \lambda (|S| + |T|).\]
    Let us prove \ref{item:permutation}. We pick edges greedily. At each step, when edge $e = \{u, v\}$ is picked, we throw away all vertices that are connected either to $u$ or $v$. At each step, at most $2d - 1$ vertices are used, so there will be at least $\tfrac{\gamma N}{2d} \geq \tfrac{N}{4d}$ steps and we get at least $D\times D$ permutation submatrix.

    Applying the theorem, we get the desired bound.
\end{proof}
As an implication, we get an $\Omega(\log N)$ lower bound for $(N, \Theta(\log^2 N), \Theta(\log N))$-expanders:
\begin{corollary}
\label{cor:log-lb-for-expanders}
    There exists an explicit graph $G$ with $N$ nodes and degree $\Theta(\log^2 N)$ such that \[\spiky(M_G) = \Omega(\log N).\]
\end{corollary}
\begin{proof}
    It is shown in \cite[Section~2.3]{Alon21} that for every prime $p$ with $p \bmod 4 = 1$ there exist two primes $q_1, q_2 = p^{\Theta(1)}$ such that for every $\ell \ge 1$ there exists a degree $p+1$ Ramanujan graph $G_\ell$ with $\Theta((q_1 q_2)^{3\ell})$ nodes, i.e. $G_\ell$ is a $(\Theta((q_1 q_2)^{3\ell}), p+1, 2\sqrt{p})$-spectral expander. Then the number of nodes in $G_\ell$ is $N_\ell = p^{\Theta(\ell)}$, so taking $\ell = \sqrt{p} / \log p$ we get that $N_\ell = 2^{\Theta(\sqrt{p})}$, so $p = \Theta(\log^2 N_\ell)$. Finally, we apply \cref{thm:lower-bound-exp} to get the desired result.
\end{proof}

We get a stronger corollary for \emph{ultra-lossless vertex expanders}. These are known to exist, but there are currently no explicit constructions. 
\begin{theorem}
\label{thm:ultralossless-lb}
    Let $G$ be a $d$-regular graph with $N$ vertices and vertex expansion $d - O(1)$ for at most $N^\delta$ vertex subsets, where $\delta \in (0, 1)$ (i.e. for any $S \subseteq V(G)$, it has at least $(d - O(1))|S|$ neighbors).
    Let $M_G$ be its adjacency matrix.
    Then $\spiky(M_G) \geq \Omega(\min(d, \log N)).$ 
\end{theorem}
\begin{proof}
    The proof is similar to the proof of \cref{thm:lower-bound-exp}, except for the proof of \ref{item:thin}. It is proved in the following way.

    Let $S\subseteq V$ with $|S| \leq N^\delta$. By the assumption of the claim, it has at least $(d - O(1))|S|$ neighbors. On the other hand the sum of degrees of vertices in $|S|$ is $d|S|$. Thus, for any $T \subseteq V(G)$ with $|T| = |S|$ we have $e(S, T) \leq d|S| - (d - O(1)) |S| + |T| = O(|S|)$.
\end{proof}

\subsection{Inner Product and Disjointness}\label{sec:ip_bound}

A simple consequence of the lower bound for random Boolean matrices (\cref{thm:all_matrices_are_hard}) is the following. Suppose an explicit $N \times N$ matrix $M$ contains a $k \times 2^k$ submatrix with all columns distinct. Then $\spiky_\pm(M) \ge \Omega\left(\tfrac{k}{\log k}\right)$. Indeed, such a submatrix necessarily contains every possible $k \times k$ Boolean matrix (up to column permutations, which do not affect $\spiky_\pm$), and hence in particular one with maximum sign spiky rank $\Omega\left(\tfrac{k}{\log k}\right)$.

The largest $k$ for which this holds is exactly the Vapnik–Chervonenkis dimension ${\rm VC}(M)$ of the matrix $M$ \cite{VC71}. Thus we obtain the general bound
\[\spiky_\pm(M) \ge \Omega\left(\frac{\VC(M)}{\log \VC(M)}\right).\]

This observation immediately implies sign spiky rank lower bounds for Inner Product and Disjointness matrices. Recall that $\IP_n(x,y) = \sum_{i \in [n]} x_i y_i \bmod 2$, where $x,y\in \{0,1\}^n$, and $\Disj_n \in \{0,1\}^{N \times N}$ is defined with $\Disj_n(x,y) \coloneqq \bigwedge_{i \in [n]} \lnot x_i \lor \lnot y_i$ where $x$ and $y$ are identified with their binary representations.
\begin{proposition}
    $\spiky_{\pm}(\IP_n)\ge \Omega(\tfrac{n}{\log n})$ and  $\spiky_{\pm}(\Disj_n) \ge \Omega(\tfrac{n}{\log n})$.
\end{proposition}
\begin{proof}
    The VC-dimension of $\IP_n$ and $\Disj_n$ is at least $n$. 
    Consider the $n \times 2^n$ submatrix of $\Disj_n$ (or $\IP_n$) consisting of the $n$ rows indexed by the standard basis vectors
\[
e_i \coloneqq (0^{i-1},1,0^{n-i}), \quad i \in [n].
\]
For a column indexed by $x \in \{0,1\}^n$, its entries in this submatrix are as follows:
\begin{itemize}
    \item In the inner product matrix $\IP_n$, the column is exactly the vector $x$.
    \item In the Disjointness matrix $\Disj_n$, the column is $1^n - x$, where $1^n$ is the all-ones vector.
\end{itemize}
Thus the submatrix captures all distinct column patterns, corresponding either to $x$ (for $\IP_n$) or its complement (for $\Disj_n$).

\end{proof}

\section{Relations to other measures}\label{sec:relations}
\subsection{Blocky Cover}
We start with a very simple connection that we will subsequently use as a lemma. A \emph{blocky cover} $\bc(M)$ of a Boolean matrix $M$ is the minimum number of blocky matrices that cover all $1$-entries of $M$.
\begin{lemma}
\label{lem:bc-to-br}
    For every $M \in \{0,1\}^{n \times n}$ we have $\spiky(M) \le \br(M) \le 2^{\bc(M)}$.
\end{lemma}
\begin{proof}
    Consider the blocky cover representation of $M = \bigvee_{i \in [\bc(M)]} A_i$. For every set $$S = \{s_1, \dots, s_\ell\} \subseteq [\bc(M)]$$ define $A_S \coloneqq A_{s_1} \circ \dots \circ A_{s_\ell}$ to be the entrywise product of $A_i$ for $i \in S$. Then by the standard inclusion-exclusion principle we have $M = \sum_{S \subseteq [\bc(M)]} (-1)^{|S|} A_S$. Since entrywise-product of blocky matrices is a blocky matrix, this is implies the claimed bound on the blocky rank.
\end{proof}

\subsection{Blocky rank}
\label{sec:sprandbr} 
In this section, we prove the following theorem.
\thmbrlefspr* 
    We call a subset $S \subseteq [N] \times [M]$ of matrix entries \emph{$t$-blocky complement} if  $\bc(\mathbb{1}_{[N]\times[M] \setminus S}) \le t$. 

    The proof proceeds as follows. We first partition the matrix entries into $2^{\spiky(M)}$ classes, according to the subset of spiky matrices in the representation that influence each entry. Each class can be written as a union of subsets of $[N]\times [N]$ with disjoint rows and columns, where each such subset is a $\spiky(M)$-blocky complement. Finally, observing that each of these subsets can be completed to a real matrix of $\spiky(M)$-rank, we obtain the desired result.
    
\begin{proof}[Proof of \cref{thm:brlefspr}]
    Let $M = R_1 + R_2 + \ldots + R_k$ be the spiky representation of $M$ with $k=\spiky(M)$. For each entry $(j,\ell)$ of $M$, we assign a vector $v_{j\ell} \in \{0, 1\}^k$ as follows: $(v_{j\ell})_i = 1$ if $(R_i)_{j\ell}\neq 0$. So, all the entries are divided into $2^k$ types based on the corresponding vector. 

   We denote $M{\uparrow}v$ a matrix obtained from $M$ by zeroing out all entries where $v$ is not assigned, so $M = \sum_{v \in \{0,1\}^k} M{\uparrow}v$.
   We will prove that $\br(M{\uparrow}v)=k^{O(k)}$ for $v \in \{0, 1\}^k$. It implies that $\br(M) = 2^k\cdot k^{O(k)}=k^{O(k)}$ by the subadditivity of blocky rank.

    Fix some $v \in \{0,1\}^k$. Let $P \subseteq [k]$ be the set of $1$-entries in $v$, and $Z \subseteq [k]$ be the set of its $0$-entries.
    We first observe that the support of $M{\uparrow}v$ is contained in the intersection of supports of $R_i$ for $i\in P$. Let $M'$ be the result of zeroing-out elements of $M$ outside of this intersection. In order to obtain $M{\uparrow}v$ from $M'$ one needs to additionally zero-out the entries covered by the union of supports of $R_i$ for $i \in Z$, see \cref{fig:support} for an example. Notice that the set of non-zero elements in $M'$ forms a blocky matrix. 
\begin{figure}[h!]
    \centering
    \scalebox{0.7}{\definecolor{color1}{HTML}{B5E5CF}
\definecolor{color2}{HTML}{FCB5AC}
\begin{tikzpicture}[scale=0.4]
\fill[white] (0,0) rectangle (10,10);
\draw[black, thick] (0,0) rectangle (10,10);
\fill[color1,opacity=0.5] (0,3) rectangle (7,10);\draw[black, thin] (0,3) rectangle (7,10);
\fill[color1,opacity=0.5] (7,0) rectangle (10,3);\draw[black, thin] (7,0) rectangle (10,3);
\fill[color1,opacity=0.5] (0,5) rectangle (6,10);\draw[black, thin] (0,5) rectangle (6,10);
\fill[color1,opacity=0.5] (6,0) rectangle (10,5);\draw[black, thin] (6,0) rectangle (10,5);
\fill[pattern=north east lines, pattern color=black]  (0,5) rectangle (6,10);
\fill[pattern=north east lines, pattern color=black]  (6,3) rectangle (7,5);
\fill[pattern=north east lines, pattern color=black]  (7,0) rectangle (10,3);

\fill[white] (12,0) rectangle (22,10);
\draw[black, thick] (12,0) rectangle (22,10);
\fill[color2,opacity=0.5] (12,0) rectangle (13,1);\draw[black, thin] (12,0) rectangle (13,1);
\fill[color2,opacity=0.5] (13,1) rectangle (14,4);\draw[black, thin] (13,1) rectangle (14,4);
\fill[color2,opacity=0.5] (14,4) rectangle (17,6);\draw[black, thin] (14,4) rectangle (17,6);
\fill[color2,opacity=0.5] (17,6) rectangle (20,8);\draw[black, thin] (17,6) rectangle (20,8);
\fill[color2,opacity=0.5] (20,8) rectangle (22,9);\draw[black, thin] (20,8) rectangle (22,9);
\fill[color2,opacity=0.5] (22,9) rectangle (22,10);\draw[black, thin] (22,9) rectangle (22,10);
\fill[color2,opacity=0.5] (12,0) rectangle (13,3);\draw[black, thin] (12,0) rectangle (13,3);
\fill[color2,opacity=0.5] (13,3) rectangle (15,6);\draw[black, thin] (13,3) rectangle (15,6);
\fill[color2,opacity=0.5] (15,6) rectangle (18,9);\draw[black, thin] (15,6) rectangle (18,9);
\fill[color2,opacity=0.5] (18,9) rectangle (19,10);\draw[black, thin] (18,9) rectangle (19,10);
\fill[color2,opacity=0.5] (19,10) rectangle (22,10);\draw[black, thin] (19,10) rectangle (22,10);
\fill[color2,opacity=0.5] (22,10) rectangle (22,10);\draw[black, thin] (22,10) rectangle (22,10);
\fill[color2,opacity=0.5] (22,10) rectangle (22,10);\draw[black, thin] (22,10) rectangle (22,10);

\fill[white] (24,0) rectangle (34,10);
\draw[black, thick] (24,0) rectangle (34,10);
\fill[color1,opacity=0.5] (24,3) rectangle (31,10);\draw[black, thin] (24,3) rectangle (31,10);
\fill[color1,opacity=0.5] (31,0) rectangle (34,3);\draw[black, thin] (31,0) rectangle (34,3);
\fill[color1,opacity=0.5] (24,5) rectangle (30,10);\draw[black, thin] (24,5) rectangle (30,10);
\fill[color1,opacity=0.5] (30,0) rectangle (34,5);\draw[black, thin] (30,0) rectangle (34,5);
\fill[color2,opacity=0.5] (24,0) rectangle (25,1);\draw[black, thin] (24,0) rectangle (25,1);
\fill[color2,opacity=0.5] (25,1) rectangle (26,4);\draw[black, thin] (25,1) rectangle (26,4);
\fill[color2,opacity=0.5] (26,4) rectangle (29,6);\draw[black, thin] (26,4) rectangle (29,6);
\fill[color2,opacity=0.5] (29,6) rectangle (32,8);\draw[black, thin] (29,6) rectangle (32,8);
\fill[color2,opacity=0.5] (32,8) rectangle (34,9);\draw[black, thin] (32,8) rectangle (34,9);
\fill[color2,opacity=0.5] (34,9) rectangle (34,10);\draw[black, thin] (34,9) rectangle (34,10);
\fill[color2,opacity=0.5] (24,0) rectangle (25,3);\draw[black, thin] (24,0) rectangle (25,3);
\fill[color2,opacity=0.5] (25,3) rectangle (27,6);\draw[black, thin] (25,3) rectangle (27,6);
\fill[color2,opacity=0.5] (27,6) rectangle (30,9);\draw[black, thin] (27,6) rectangle (30,9);
\fill[color2,opacity=0.5] (30,9) rectangle (31,10);\draw[black, thin] (30,9) rectangle (31,10);
\fill[color2,opacity=0.5] (31,10) rectangle (34,10);\draw[black, thin] (31,10) rectangle (34,10);
\fill[color2,opacity=0.5] (34,10) rectangle (34,10);\draw[black, thin] (34,10) rectangle (34,10);
\fill[color2,opacity=0.5] (34,10) rectangle (34,10);\draw[black, thin] (34,10) rectangle (34,10);
\fill[pattern=north east lines, pattern color=black] (24,9) rectangle (30,10);
\fill[pattern=north east lines, pattern color=black] (24,8) rectangle (27,9);
\fill[pattern=north east lines, pattern color=black] (24,6) rectangle (27,8);
\fill[pattern=north east lines, pattern color=black] (24,5) rectangle (25,6);
\fill[pattern=north east lines, pattern color=black] (29,5) rectangle (30,6);
\fill[pattern=north east lines, pattern color=black] (30,3) rectangle (31,5);
\fill[pattern=north east lines, pattern color=black] (31,0) rectangle (34,3);

\node at (5,-1.5) {$R_1 \cap R_2$};
\node at (17,-1.5) {$R_3 \cup R_4$};
\node at (29,-1.5) {$(R_1 \cap R_2) \setminus (R_3 \cup R_4)$};

\end{tikzpicture}}
    \caption{The picture illustrates the process for $v=(1100)$, so $P = \{1,2\}$ and $Z = \{3,4\}$. Here the green matrices illustrate $R_{1}$ and $R_2$, the red matrices are $R_3$ and $R_4$, the hatched area in the right square is the support of the matrix $M{\uparrow}v$.}
    \label{fig:support}
\end{figure}
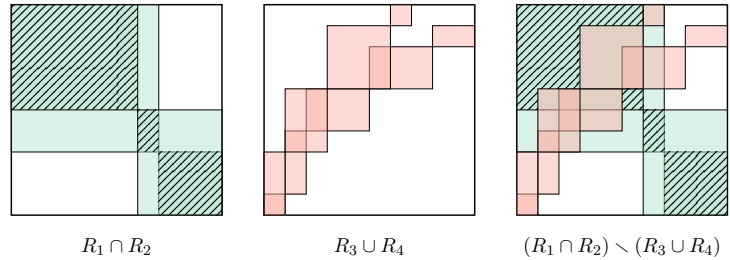

    Let $D$ be a block in $M'$. It suffices to upper bound $\br((M{\uparrow}v)|_D)$, since $\br(M{\uparrow}v)$ is the maximum of the blocky ranks in the blocks. Let $L = \sum_{i \in P} R_i|_D$, $S$ be the set of entries of $D$ that are non-zero in $M{\uparrow}v$, by definition $S$ is a $k$-blocky complement. Then $L$ agrees with $M|_D$ on $S$. This means $L|_S$ is Boolean. The rank of $L$ can be upper bounded as:
\[\rk(L) = \rk(\sum_{i \in P} R_i|_D) \le \sum_{i \in P} \rk(R_{i}|_D) \le k\]
    where $\rk(R_{i}|_D) \leq 1$ for each $i \in P$ since $D$ is a subset of a block in $R_{i}$. To proceed further, we need the following lemma:
    \begin{lemma}\label{lem:brcompl}
    Let $M$ be a real-valued matrix with $\rk(M)\leq r$ and a subset $S$ of its entries be a $t$-blocky complement, such that $M|_S$ is Boolean. Then there exist $m$ blocky matrices $B_1, \dots, B_m$, and coefficients $\alpha_1, \dots, \alpha_m \in \{\pm 1\}$, such that $m = (t+r)^{O(t + r)}$, and $\sum_{i = 1}^m\alpha_i B_i$ agrees with $M$ on $S$: for all $(j, \ell) \in S$ 
    $$\Big(\sum_{i=1}^m\alpha_i B_i\Big)_{j\ell} = M_{j\ell}.$$
    \end{lemma}

    Now, \cref{lem:brcompl} can be applied to the matrix $L$, and we get $\sum_{i=1}^m \alpha_i B_i$ agrees with $L$ on $S$. Observe that this sum is not necessarily Boolean: we only control its entries in $S$. Now we will get blocky representation of $(M{\uparrow}v)|_D$ based on $\sum_{i=1}^m\alpha_i B_i$: we need to zero-out the entries outside of $S$. 

     By definition $(M{\uparrow}v)|_D = \sum_{i=1}^m (B_i \circ \mathbb{1}_{\overline{S}})$. By \cref{lem:bc-to-br} we have $\br(\mathbb{1}_{\overline{S}}) \le 2^k$ since $S$ is a $k$-blocky complement. Then since entry-wise product of two blocky matrices is a blocky matrix we get $\br(B_i \circ \mathbb{1}_{\overline{S}}) \le 2^k$ for each $i$ and by subadditivity of blocky rank we get $\br((M{\uparrow}v)|_D) \le m \cdot 2^k \le k^{O(k)}$ and by combining these across all blocks $D$ in $M'$ we get the same bound for $\br(M{\uparrow}v)$.
\end{proof}

Now it remains to prove the lemma:
\begin{proof}[Proof of \cref{lem:brcompl}]
        For simplicity, we assume $\rk(M) = r$. The case of a smaller $\rk(M)$ follows immediately. We prove that $m$ can be bounded by $T_t$, where $T_t = rt\cdot T_{t-1}+2^r$ with $T_0 = 2^r$. We prove it by induction on $t$.
    \paragraph{\boldmath Base case: $t = 0$.}

        In this case, $M$ is a Boolean matrix of rank at most $r$. Let $c_1, c_2, \ldots, c_r$ be the indices of columns forming the basis. Thus, every row is fully determined by its entries at positions $c_1, c_2, \ldots, c_r$. So, there can be at most $2^r$ different rows in $M$ and thus $\br(M)\leq 2^r$, as we can represent each type of rows with one rectangle.

     \paragraph{\boldmath Induction step: $t > 0$.}

        Let $L_1, \ldots, L_t$ be the blocky matrices, whose union is $\overline{S}$, and $N_{i1}, \ldots , N_{it_i}$ be blocks of $L_i$ for each $i \in [t]$. We denote with $R$ the set of rows in $M$, and with $C$ the set of columns in $M$.

        Let $c_1, c_2, \ldots, c_r \in C$ be the indices of columns forming the basis. We call a row \emph{spoiled} if there exists $\ell$, such that its $c_\ell$-th entry belongs to $N_{ij}$ for some $i, j$. We assign each of the spoiled rows to the corresponding class $R_{ij}$. If one row corresponds to several classes, we choose one of them arbitrarily. Note that there are at most $rt$ non-empty classes, because each column intersects at most one of the blocks $N_{i1}, N_{i2}, \ldots, N_{it_i}$ for every $i$. Now we will solve the problem separately for each (non-empty) $R_{ij}$ and for $R\setminus \bigcup_{i, j}R_{ij}$.

        First, we use the induction hypothesis for $M|_{R_{ij}\times (C\setminus C(N_{ij}))}$ and $S \cap (R_{ij}\times(C\setminus C(N_{ij})))$, where $C(N_{ij})$ is set of columns forming the block $N_{ij}$. Note that $S \cap (R_{ij}\times(C\setminus C(N_{ij})))$ is $(t-1)$-blocky complement wrt $M|_{R_{ij}\times (C\setminus C(N_{ij}))}$, because $L_i$ is excluded from the complement. Thus, we get exactly the blocky decomposition which agrees with $M|_{R_{ij}\times C}$ on $S\cap (R_{ij}\times C)$.

        Second, we use the same argument as in the induction base to prove that $\br(M|_{(R \setminus \bigcup_{ij}R_{ij}) \times C}) \leq 2^r$.

        Overall, we get blocky decomposition, which agrees with $M$ on $S$ with at most $T_t = rt\cdot T_{t-1} + 2^r$ matrices. 

    It is easy to see using induction that
    $$T_t = r^t t!\cdot T_0 + 2^r\sum_{i = 0}^{t-1}r^i \prod_{j=0}^{i-1}(t-j). $$
    Thus, we can upper bound $T_t$ with
    $T_t \leq2^{O(t\log r + t\log t + r)} + t\cdot 2^{O(r + t\log r +t\log t)}\leq (r+t)^{O(r+t)}.$
\end{proof}

    We observe that without the assumption on the structure of $S$, the lemma is completely false.
\begin{proposition} 
    For every $n \in \mathbb{Z}_{> 0}$ there exists $N$ and a matrix $M \in \mathbb{R}^{N \times N}$ of rank $2$, such that deterministic communication complexity of computing $M(x,y)$ with the promise that $M(x,y) \in \{0,1\}$ is at least $n$.
\end{proposition}
\begin{proof}
Consider the following communication problem $\mathcal{M}$ that embeds all rank-$2$ matrices: Alice's input is $(x_1,x_2) \in \mathbb{R}^2$, Bob's input is $(y_1,y_2) \in \mathbb{R}^2$. They need to return $x_1 x_2 + y_1 y_2$ whenever this value is in $\{0,1\}$. We will show that the communication cost of this problem is $\omega(1)$ which implies that for every $n$ there is a finite $N \times N$ submatrix of $\mathcal{M}$ that satisfies the claim.

If there is a cost-$k$ protocol for a communication problem, then there is a cost-$2^{k+1}$ one-sided protocol for that problem, since Alice can just send her answers for all potential replies of Bob. Hence, if there exists an $O(1)$-cost protocol for $\mathcal{M}$, then there is a partition $A_1 \sqcup A_2 \sqcup \dots \sqcup A_k = \mathbb{R}^2$ such that for each $A_i$ and $(x,y) \in \mathbb{R}^2$ there exists $b \in \{0,1\}$ such that $\langle z, (x,y) \rangle \neq b$ for every $z \in A_i$ (in that case, Alice can send $i$ and Bob can determine the answer by his input $(x,y)$).

So let $(\mathbb{R}^2, \{(u,v) \in \binom{\mathbb{R}^2}{2} \mid \exists t \in \mathbb{R}^2\colon \{\langle u,t \rangle, \langle v,t \rangle\} = \{0,1\}\})$ be a graph, the partition above is exactly $k$-coloring of this graph. We now claim that every set of points $u_1, \dots, u_n \in \mathbb{R}^2 \setminus \{0\}$ such that $\{0,u_1,\dots,u_n\}$ are in a general position, forms a clique in this graph contradicting that it is $O(1)$-colorable. Indeed, for every pair of points $u_i, u_j$ the system $\langle u_i, t\rangle = 1$ and $\langle u_j, t\rangle = 0$ has a solution: take $t' = (-(u_j)_1, (u_j)_0)$, so the second equation is satisfied by any $\alpha t'$ and then find appropriate $\alpha$ to satisfy the first equation, it fails only in the case if $\{0,u_i,u_j\}$ are on a line, so since $\{0,u_1,\dots, u_n\}$ are in the general position, $\{u_1,\dots, u_n\}$ is indeed a clique.
\end{proof}

\subsection{Sparsity}\label{sec:sparsity}

\begin{lemma}[\cref{lem:rigidity_vs_sparsity} restated]
Let $M \in \F^{N \times N}$ be a matrix over an arbitrary field $\F$ \footnote{we think of $\F$ either $\R$ or a finite field $\F_p$.} and let $\spar(M) \leq m$. Then
$$\spiky^{\F}(M) \leq 2\sqrt{m}.$$ 
\end{lemma}
\begin{proof}
    Let $c_1, \dots, c_N$ denote the number of non-zero entries in each column of $M$, so $\sum_{i=1}^N c_i = m$. 
    We show a process to decompose $M$ into spiky rank-$1$ matrices.
    
    Note that any matrix with a single non-zero column has spiky rank 1.
    Also, a matrix with at most one non-zero entry in each column has spiky rank 1. 

    We now describe a strategy that decomposes $M$ into matrices of these two types. Consider the following game on the column counts $c_1, \ldots, c_N$:
    \begin{description}[noitemsep]
        \item[Type 1 move:] nullify a single $c_i$ (i.e., remove a matrix with one non-zero column).
        \item[Type 2 move:] subtract $1$ from each $c_i$ (i.e., remove a matrix with one non-zero entry per column).
    \end{description}
Each such move corresponds to adding a spiky rank-1 matrix to our decomposition, contributing 1 to the spiky rank. Our goal is to make all $c_i$'s zero using as few moves as possible. Here is a simple strategy that makes only $2\sqrt{m}$ moves:

\begin{enumerate}[noitemsep]
    \item While there exists $i$ such that $c_i > \sqrt{m}$, perform a type 1 move on that column. Since the total sum of $c_i$'s is at most $m$, there can be at most $\sqrt{m}$ such $i$. Thus, this phase requires at most $\sqrt{m}$ moves.
    \item Now $c_i \leq \sqrt{m}$ for all $i$. Perform type 2 moves repeatedly. After at most $\sqrt{m}$ such steps, all $c_i$ become zero.
\end{enumerate}

Altogether, we use at most $2\sqrt{m}$ moves, each corresponding to a spiky rank-1 matrix. Therefore, $\spiky^{\F}(M) \leq 2\sqrt{m}$.
\end{proof}

We will further improve \cref{lem:rigidity_vs_sparsity} for Boolean matrices using the recipe from \cite[Theorem~2]{AvrahamYehudayoff2022} (preprint version of \cite{DY24}). We will need the following simple technical lemma:

\begin{lemma}[Improved Lemma~11~from~\cite{AvrahamYehudayoff2022}]  \label{lem:ay24lem11}
    For large enough $N$, every $M \in \{0,1\}^{N \times N}$  with $\spar(M) \geq \frac{N^2}{\log^6 N}$ contains $\frac{\log N}{100\log\log N}\times \sqrt{N}$ all-$1$ rectangle.
\end{lemma}
\begin{proof}
The reference we use is Lemma 5.4 in \cite{RaoYehudayoffBook}, which says that if $\varepsilon>0$ and $k$ are such that
\[
   \varepsilon N \ge 2k
   \quad\text{and}\quad
   k \le \frac{\log N}{2\log\bigl(2e/\varepsilon\bigr)},
\]
then every matrix $M$ with $\spar(M)\ge \varepsilon N^2$ contains a $1$-submatrix with at least $k$ rows and $\sqrt{n}$ columns.  Applying this lemma with
\(
   \varepsilon = (1000\log^6 N)^{-1}
\)
completes the proof.
\end{proof}

\begin{theorem}\label{thm:sparsity_vs_br_Boolean}
    Any Boolean matrix $M$ with sparsity $m$ satisfies
    $$\br(M)\leq O\left(\frac{\sqrt{m}\log\log m}{\log m}\right).$$ 
\end{theorem}
\begin{proof}
    We are going to eliminate all one-entries of $M$ by subtracting blocky matrices from $M$ and maintain the property that after the subtraction the resulting matrix is still Boolean. Thus, the blocky decomposition we build is a \emph{blocky partition}. At each step, we assume that each row and each column of $M$ contains a nonzero element. Otherwise, the empty rows and columns can be removed from the matrix with no change to the blocky rank.

    First, we adopt the same strategy as in \cref{lem:rigidity_vs_sparsity} for the matrices $M$, s.t. its number of rows or number of columns is at least $\sqrt{m}\log m$. We can do at most $\frac{m}{\sqrt{m}\log m}$ such steps, using one blocky matrix at each step, thus using at most $\frac{\sqrt{m}}{\log m}$ blocky matrices.

    Let $r(A)$ be the number of (non-empty) rows in $A$, and $c(A)$ be the number of (non-empty) columns in $A$.
    So, by the above, we may assume that $r(M) \leq \sqrt{m}\log m$ and $c(M) \leq \sqrt{m}\log m$. We now repeat the proof of \cite[Theorem 2]{AvrahamYehudayoff2022}, ensuring that the same strategy works in our case.

    We construct a nested matrix sequence $$M = M_0 \supseteq M_1 \supseteq M_2 \supseteq \ldots\supseteq M_N,$$ with controlled $\br(M_i-M_{i+1})$ and $\br(M_N) \leq \frac{\sqrt{m}}{\log m}$. 

    Assume we have constructed $M_0, M_1, \ldots, M_t$. If we reach a situation where $c(M_t)\leq \sqrt{m}$ and $r(M_t)\leq \sqrt{m}$, we set $N=t$ and conclude by \cite[Theorem~1.4]{DY24}. Otherwise, we use the following algorithm.

    Let $$k = \frac{\log m}{100 \log\log m} \quad\text{and}\quad \ell = \sqrt[4]{m}.$$

    Without loss of generality, suppose $c(M_t) \geq \sqrt{m}$. We construct the matrix $M_{t+1}$ with $c(M_{t+1})\leq c(M_t)/2$. Assume we have already constructed $B_0, B_1, \ldots, B_q$ and want to construct $B_{q+1}$, with $B_0 = 0$. Let
    $$M' = M_t-\sum_{i=0}^q B_i.$$

    Let us run the following procedure:

    \begin{algorithm}[ht]
    \caption{Extracting a Sequence of Submatrices}
    \begin{algorithmic}[1]
      \State $j \gets 0$
      \State $N_0 \gets M'$
      \While{$\spar(N_j) \ge m / \log^4 m$} \Comment{by \cref{lem:ay24lem11} with $n = \sqrt{m}\,\log m$}
        \State Decompose $N_j$ as
        \[
          N_j \;=\;
          \begin{pmatrix}
            1_{k\times \ell} & * \\
            *                 & N_{j+1}
          \end{pmatrix}
        \]
        \State $j \gets j + 1$
      \EndWhile
    \end{algorithmic}
    \end{algorithm}

    Let $T$ denote the final value of $t$. Let $B$ be the blocky matrix with the $T$ one-blocks of size $k \times \ell$ that were found above. It could sometimes be the zero matrix. Define $L$ with
    \[
    L \coloneqq
    \begin{pmatrix}
    0 & 0 \\
    0 & N_{T}
    \end{pmatrix}.
    \]

    The matrices $B$ and $L$ have the same size as $M_t$.

    There are two options to consider.

    \textbf{The first option} is that $T\ell \geq \frac{c(M_t)}{2}$. In this case, we set
    \(
    B_{q+1} = B
    \) and continue the procedure.
    There are at least $T\ell$ nonzero columns in $B$, and each column has $k$ nonzero entries.
    Hence
    \[
    \lvert B_{q+1} \rvert \;\geq\; \frac{c(M_t)k}{2}\;\geq \;\frac{\sqrt{m}k}{2}.
    \]

    \textbf{The second option} is that $T\ell < \frac{c(M_t)}{2}$.  
    Denote by $B'$ the matrix obtained from $M'$ by zeroing out every zero row in $B$.  
    Because the total number of nonzero rows in $B$ is at most $\frac{k\cdot c(M_t)}{\ell}$ we conclude
    \[
    \br(B') \;\le\; \frac{k\cdot c(M_t)}{\ell}\;\leq\; k \sqrt[4]{m} \log m.
    \]
    Moreover, because $\spar(L) < \frac{m}{\log^4m}$, \cref{lem:rigidity_vs_sparsity} implies
    \(
    \mathrm{\br}(L) \le \sqrt{m}/\log^2 m.
    \)
    Then the matrix \(M' - L - B'\) satisfies
    \(
    c\bigl(M' - L - B'\bigr) \le T\ell < c(M_t)/2.
    \)
    We have thus reduced the number of columns by a factor of two and we set $M_{t + 1} = M'-L-B'$. So, we have \[\br(M_{t} - M_{t + 1})\leq \br (L) + \br(B) + q \leq \frac{\sqrt{m}}{\log^2m} + k\sqrt[4]{m} \log m + \frac{2\cdot\spar(M_t-M_{t+1})}{\sqrt{m}k}.\]

    Each time either $c(M_t)$ or $r(M_t)$ is reduced $2$ times, so we will do at most $2\log\log m$ steps to reach $c(M_t) \leq \sqrt{m}$ and $r(M_t) \leq \sqrt{m}$. Thus, we get
    \begin{align*}
    \br(M)&\leq \log\log m \cdot \left(\frac{\sqrt{m}}{\log^2m} + k\sqrt[4]{m}\log m\right) + \frac{2\cdot\spar(M-M_N)}{\sqrt{m}k}\\
    &\leq o\left(\frac{\sqrt{m}}{\log m}\right) + \frac{2\cdot\spar(M)}{\sqrt{m}k}\\
    &\leq O\left(\frac{\sqrt{m}\log\log m}{\log m}\right). \qedhere
    \end{align*}
\end{proof}

\subsection{$\gamma_2$-norm}
\label{sec:gamma-2-vs-br}
In this section, we give a non-explicit quadratic separation between sign spiky rank and $\gamma_2$-norm.
\begin{theorem}[\cref{thm:gamma-2-vs-br} restated.] \label{thm:gamma-2-vs-br-restated}
    There exists a Boolean matrix $M$ such that $$\spiky(M) \ge \Omega(\gamma_2^2(M)).$$
\end{theorem}
\begin{lemma}
\label{thm:large-br-bounded-deg}
    There exists an $N$-vertex bipartite graph $G$ of left degree at most $d \le N^{1/4}$ such that its adjacency matrix has $\spiky_\pm(M_{G}) \geq d/18$.
\end{lemma}
\begin{proof}
    The number of left-regular graphs with left degree $d$ is exactly $\binom{N}{d}^N$. On the other hand, the number of boolean matrices with sign spiky rank $r$ is at most $N^{6Nr}$ by \cref{lem:counting_spr}. It is then sufficient to have 
    \begin{equation}
    \label{eq:requirement}
        \binom{N}{d}^N > N^{6Nr} 
    \end{equation}
    to conclude that there exists a left-degree-$d$ graph $G$ with adjacency matrix $M_G$ such that $$\spiky_{\pm}(M_G) > r.$$ Choose $r = d/18$, then
    \( N^{6Nr} = N^{Nd/3}.\) On the other hand $\binom{N}{d}^N \ge N^{Nd} / (2d)^{Nd}$, so by rearranging \eqref{eq:requirement} is equivalent to
    \( N^{Nd - Nd / 3} > (2d)^{Nd}, \)
    which is satisfied if $d \le N^{1/4}$.
\end{proof}
\begin{proof}[Proof of \cref{thm:gamma-2-vs-br-restated}]
    Since $\ell_2$-norm of the rows of adjacency matrix of a left-$d$-regular graph is $\sqrt{d}$ we get that for an adjacency matrix $M$ of such a graph we have $\gamma_2(M) \le \sqrt{d}$. Hence, for a graph given by \cref{thm:large-br-bounded-deg} we have that $\br(M) \ge \spiky(M) \ge \spiky_{\pm}(M) \ge \Omega(\gamma_2^2(M))$. 
\end{proof}

Can this separation be made explicit? A variation on \cref{thm:lower-bound-exp} may offer a way to achieve it: it says that for an \emph{ultra-lossless} degree-$d$ expander $G$ we have $\spiky(M_G) = \Omega(d)$. Since $\gamma_2(M_G) \le \|M_G\|_{\row} \le \sqrt{d}$ for every $d$-regular graph $G$, it is another way to prove \cref{thm:gamma-2-vs-br} non-explicitly (albeit only for $\spiky$ and not for $\spiky_\pm$). The expanders required in \cref{thm:ultralossless-lb} are known to exist, but there are currently no explicit constructions, \cite{HLMRZ25} poses constructing these expanders as an open problem. 

Observe that $\HDone{n}$ shows a separation between approximate $\gamma_2$-norm and blocky rank, because approximate $\gamma_2$-norm of $\HDone{n}$ is $O(1)$ whereas its blocky rank at least $\sqrt{n}$ by~\cref{thm:lower-bound-hd1}.

\section*{Acknowledgments} 
We thank Nathan Harms for granting permission to include his proof of the upper bound in \cref{thm:sbr-upper-bound}, as well as for helpful feedback that improved the presentation of the paper. We are grateful to Dmitry Sokolov for bringing ultra-lossless expanders to our attention, and to Amir Yehudayoff for pointing us to Warren's theorem for the proof of \cref{thm:all_matrices_are_hard}. Finally, we thank Nati Linial, Shachar Lovett, Ryan Williams, and Amir Yehudayoff for many insightful discussions on spiky rank.

\bibliographystyle{alpha}
\bibliography{bib}

\newcommand{\etalchar}[1]{$^{#1}$}
\begin{thebibliography}{BBM{\etalchar{+}}21}

\bibitem[Alo86]{alon1986covering}
Noga Alon.
\newblock Covering graphs by the minimum number of equivalence relations.
\newblock {\em Combinatorica}, 6(3):201--206, 1986.

\bibitem[Alo21]{Alon21}
Noga Alon.
\newblock Explicit expanders of every degree and size.
\newblock {\em Combinatorica}, 41(4):447--463, 2021.

\bibitem[AMY16]{alon2016sign}
Noga Alon, Shay Moran, and Amir Yehudayoff.
\newblock Sign rank versus vc dimension.
\newblock In {\em Conference on Learning Theory}, pages 47--80. PMLR, 2016.

\bibitem[AN25]{applebaum2025meta}
Benny Applebaum and Oded Nir.
\newblock The meta-complexity of secret sharing.
\newblock In {\em Proceedings of the 57th Annual ACM Symposium on Theory of Computing}, pages 965--976, 2025.

\bibitem[AW17]{AW17}
Josh Alman and Ryan Williams.
\newblock Probabilistic rank and matrix rigidity.
\newblock In {\em Proceedings of the 49th Annual ACM SIGACT Symposium on Theory of Computing}, STOC 2017, page 641–652, New York, NY, USA, 2017. Association for Computing Machinery.

\bibitem[AY22]{AvrahamYehudayoff2022}
Daniel Avraham and Amir Yehudayoff.
\newblock On blocky ranks of matrices.
\newblock Technical Report TR22-137, Electronic Colloquium on Computational Complexity (ECCC), September 2022.

\bibitem[AY24]{DY24}
Daniel Avraham and Amir Yehudayoff.
\newblock On blocky ranks of matrices.
\newblock {\em computational complexity}, 33, 03 2024.

\bibitem[BBM{\etalchar{+}}21]{blockP4free2021}
Alexander~R. Block, Simina Br{\^a}nzei, Hemanta~K. Maji, Himanshi Mehta, Tamalika Mukherjee, and Hai~H. Nguyen.
\newblock {{P}}{$_4$}-free partition and cover numbers \& applications.
\newblock In Stefano Tessaro, editor, {\em 2nd {{Conference}} on {{Information-Theoretic Cryptography}} ({{ITC}} 2021)}, volume 199 of {\em Leibniz {{International Proceedings}} in {{Informatics}} ({{LIPIcs}})}, pages 16:1--16:25, Dagstuhl, Germany, 2021. Schloss Dagstuhl -- Leibniz-Zentrum f{\"u}r Informatik.

\bibitem[BHT25]{balla2025factorization}
Igor Balla, Lianna Hambardzumyan, and Istv{\'a}n Tomon.
\newblock Factorization norms and an inverse theorem for maxcut.
\newblock {\em arXiv preprint arXiv:2506.23989}, 2025.

\bibitem[CHZZ24]{CHZZ24}
Tsun~Ming Cheung, Hamed Hatami, Rosie Zhao, and Itai Zilberstein.
\newblock Boolean functions with small approximate spectral norm.
\newblock {\em Discrete Analysis}, sep 18 2024.

\bibitem[CMS20]{CMS20}
Arkadev Chattopadhyay, Nikhil~S. Mande, and Suhail Sherif.
\newblock The log-approximate-rank conjecture is false.
\newblock {\em Journal of the ACM}, 67, 2020.

\bibitem[DGJ{\etalchar{+}}10]{DGJSV10}
Ilias Diakonikolas, Parikshit Gopalan, Ragesh Jaiswal, Rocco~A. Servedio, and Emanuele Viola.
\newblock Bounded independence fools halfspaces.
\newblock {\em SIAM Journal on Computing}, 39(8):3441--3462, 2010.

\bibitem[Fra82]{frankl1982covering}
P\'eter Frankl.
\newblock Covering graphs by equivalence relations.
\newblock In {\em North-Holland Mathematics Studies}, volume~60, pages 125--127. Elsevier, 1982.

\bibitem[GG12]{GG73}
Martin Golubitsky and Victor Guillemin.
\newblock {\em Stable mappings and their singularities}, volume~14.
\newblock Springer Science \& Business Media, 2012.

\bibitem[GH25]{goh2025block}
Marcel~K Goh and Hamed Hatami.
\newblock Block complexity and idempotent schur multipliers.
\newblock {\em International Mathematics Research Notices}, 2025(24):356, 2025.

\bibitem[GHIS25]{GHIS25}
Mika G{\"o}{\"o}s, Nathaniel Harms, Valentin Imbach, and Dmitry Sokolov.
\newblock Sign-rank of $ k $-hamming distance is constant.
\newblock {\em arXiv preprint arXiv:2506.12022}, 2025.

\bibitem[GHR25]{goos2025equality}
Mika G{\"o}{\"o}s, Nathaniel Harms, and Artur Riazanov.
\newblock Equality is far weaker than constant-cost communication.
\newblock {\em arXiv preprint arXiv:2507.11162}, 2025.

\bibitem[Gri80]{grigoriev}
Dima Grigoriev.
\newblock An application of separability and independence notions for proving lower bounds of circuit complexity.
\newblock {\em J.Soviet Math.}, 14:1450--1456, 01 1980.

\bibitem[HHH23]{HHH23}
Lianna Hambardzumyan, Hamed Hatami, and Pooya Hatami.
\newblock Dimension-free bounds and structural results in communication complexity.
\newblock {\em Israel Journal of Mathematics}, 253(2):555--616, 2023.

\bibitem[HLM{\etalchar{+}}25]{HLMRZ25}
Jun-Ting Hsieh, Alexander Lubotzky, Sidhanth Mohanty, Assaf Reiner, and Rachel~Yun Zhang.
\newblock Explicit lossless vertex expanders.
\newblock {\em arXiv preprint arXiv:2504.15087}, 2025.

\bibitem[HLW06]{hoory06}
Shlomo Hoory, Nathan Linial, and Avi Wigderson.
\newblock Expander graphs and their applications.
\newblock {\em Bull. Amer. Math. Soc.}, 43(04):439–562, August 2006.

\bibitem[HP10]{hansen2010exact}
Kristoffer~Arnsfelt Hansen and Vladimir~V Podolskii.
\newblock Exact threshold circuits.
\newblock In {\em 2010 IEEE 25th Annual Conference on Computational Complexity}, pages 270--279. IEEE, 2010.

\bibitem[Juk06]{jukna2006graph}
Stasys Jukna.
\newblock On graph complexity.
\newblock {\em Combinatorics, Probability and Computing}, 15(6):855--876, 2006.

\bibitem[KSP20]{KSP20}
Rohan Karthikeyan, Siddharth Sinha, and Vallabh Patil.
\newblock On the resolution of the sensitivity conjecture.
\newblock {\em Bulletin of the American Mathematical Society}, 57:1, 03 2020.

\bibitem[KW16]{Kane_Williams2016}
Daniel~M Kane and Ryan Williams.
\newblock Super-linear gate and super-quadratic wire lower bounds for depth-two and depth-three threshold circuits.
\newblock In {\em Proceedings of the forty-eighth annual ACM symposium on Theory of Computing}, pages 633--643, 2016.

\bibitem[LS{\etalchar{+}}09]{lee2009lower}
Troy Lee, Adi Shraibman, et~al.
\newblock Lower bounds in communication complexity.
\newblock {\em Foundations and Trends{\textregistered} in Theoretical Computer Science}, 3(4):263--399, 2009.

\bibitem[MB17]{mukherjee2017lower}
Anirbit Mukherjee and Amitabh Basu.
\newblock Lower bounds over boolean inputs for deep neural networks with relu gates.
\newblock {\em arXiv preprint arXiv:1711.03073}, 2017.

\bibitem[PR94]{PudlakRodl}
P.~Pudlák and V.~Rödl.
\newblock Some combinatorial-algebraic problems from complexity theory.
\newblock {\em Discrete Mathematics}, 136(1):253--279, 1994.

\bibitem[PSS23]{pitassiStrength2023}
Toniann Pitassi, Morgan Shirley, and Adi Shraibman.
\newblock The strength of equality oracles in communication.
\newblock In Yael Tauman~Kalai, editor, {\em 14th {{Innovations}} in {{Theoretical Computer Science Conference}} ({{ITCS}} 2023)}, volume 251 of {\em Leibniz {{International Proceedings}} in {{Informatics}} ({{LIPIcs}})}, pages 89:1--89:19, Dagstuhl, Germany, 2023. Schloss Dagstuhl -- Leibniz-Zentrum f{\"u}r Informatik.

\bibitem[Raz89]{razborovRigidity}
Alexander Razborov.
\newblock On rigid matrices (in russian).
\newblock 1989.

\bibitem[ROS94]{roychowdhury1994lower}
Vwani~P Roychowdhury, Alon Orlitsky, and Kai-Yeung Siu.
\newblock Lower bounds on threshold and related circuits via communication complexity.
\newblock {\em IEEE Transactions on Information Theory}, 40(2):467--474, 1994.

\bibitem[RY20]{RaoYehudayoffBook}
Anup Rao and Amir Yehudayoff.
\newblock {\em Communication Complexity and Applications}.
\newblock January 2020.
\newblock Publisher Copyright: {\textcopyright} Anup Rao and Amir Yehudayoff 2020.

\bibitem[SS96]{SS96}
M.~Sipser and D.A. Spielman.
\newblock Expander codes.
\newblock {\em IEEE Transactions on Information Theory}, 42(6):1710--1722, 1996.

\bibitem[Val77]{valiant1977graph}
Leslie~G Valiant.
\newblock Graph-theoretic arguments in low-level complexity.
\newblock In {\em Mathematical Foundations of Computer Science 1977: Proceedings, 6th Symposium, Tatransk{\'a} Lomnica September 5--9, 1977 6}, pages 162--176. Springer, 1977.

\bibitem[VC71]{VC71}
V.~N. Vapnik and A.~Ya. Chervonenkis.
\newblock On the uniform convergence of relative frequencies of events to their probabilities.
\newblock {\em Theory of Probability \& Its Applications}, 16(2):264--280, 1971.

\bibitem[War68]{warren1968lower}
Hugh~E Warren.
\newblock Lower bounds for approximation by nonlinear manifolds.
\newblock {\em Transactions of the American Mathematical Society}, 133(1):167--178, 1968.

\bibitem[Wil18]{williams2018limits}
Ryan Williams.
\newblock Limits on representing boolean functions by linear combinations of simple functions: Thresholds, relus, and low-degree polynomials.
\newblock {\em arXiv preprint arXiv:1802.09121}, 2018.

\bibitem[Wil24a]{williamsOrthogonal2024}
Ryan Williams.
\newblock The {{Orthogonal Vectors Conjecture}} and non-uniform circuit lower bounds.
\newblock In {\em 2024 {{IEEE}} 65th {{Annual Symposium}} on {{Foundations}} of {{Computer Science}} ({{FOCS}})}, pages 1372--1387, October 2024.

\bibitem[Wil24b]{williams2024personal_communication}
Ryan Williams.
\newblock Personal communication.
\newblock 2024.

\bibitem[Wun12]{wunderlich2012theorem}
Henning Wunderlich.
\newblock On a theorem of razborov.
\newblock {\em computational complexity}, 21:431--477, 2012.

\end{thebibliography}
\appendix 

\section{Sign and Approximate versions of Blocky rank}\label{sec:approx_sign}

\subsection{Blocky rank and Sign Blocky Rank}
\label{sec:br-vs-sign-br}
We refer to \cite{HHH23} for the definition of $\D^{\textsc{Eq}}$: deterministic communication complexity with \textsc{Equality} oracle.
\begin{theorem}
\label{thm:br-vs-sign-br}
    $\br(M) \le 2^{\br_{\pm}(M)}$.
\end{theorem}
\begin{proof}
    As shown in \cite{HHH23} we have $\D^{\textsc{Eq}}(M) \ge \log \br(M)$. On the other hand observe that $\D^{\textsc{Eq}}(M) \le \br_{\pm}(M)$: let $M = \sign(\alpha_1  B_1 + \dots + \alpha_r B_r)$ for blocky matrices $B_1, \dots, B_r$. Alice and Bob then compute $B_i(x,y)$ for each $i \in [r]$, where $x$ and $y$ are their inputs, taking $r$ \textsf{Equality} queries. After this computation both of them can compute $M(x,y) = \sign(\sum_{i \in [r]} B_i(x,y) \alpha_i)$ locally. Taking the upper and the lower bound for $D^{\textsc{Eq}}(M)$ together implies the claimed relation.
\end{proof}

\subsection{Sign Blocky Rank Bound of $\HDone{n}$}
We show a tight bound for the sign blocky rank of $\HDone{n}$.
\begin{theorem} 
\label{thm:sbr-upper-bound}
    $\sbr(\HDone{n}) = \Theta(\log n)$.
\end{theorem}
\begin{proof}
    The lower bound follows directly from \cref{thm:lower-bound-hd1} and \cref{thm:br-vs-sign-br}, so we proceed to prove the upper bound.
    Without loss of generality, suppose that $n$ is a power of $2$, otherwise increase $n$ so it is. Let $M_{i,j}[x,y]$ for $i \in [\log_2 n]$, $j \in \{0,1\}$, $x,y \in \{0,1\}^n$ be defined as follows: $M_{i,j}[x,y] = 0$ iff $x$ and $y$ coincide on all indices $t$ with the $i$th bit in the representation equal to $j$ and $1$ otherwise. For example if $M_{i,0}[x,y] = M_{i,1}[x,y]=0$ then $x = y$ and vise versa. If $x$ and $y$ differ in exactly one position $\ell \in [n]$ then for every $i \in [\log_2 n]$ we have $M_{i,0}[x,y] + M_{i,1}[x,y] = 1$. If $x$ and $y$ differ in positions $\ell_1 \neq \ell_2$, then for $i \in [\log_2 n]$ that corresponds to the difference in bit representations of $\ell_1$ and $\ell_2$ we get $M_{i,0}[x,y] + M_{i,1}[x,y] = 2$. 

    Therefore $x,y$ differ in exactly one position iff $\sum_{i\in[\log_2 n]; j \in \{0,1\}} M_{i,j}[x,y] \le k \coloneqq \log_2 n$ and $x \neq y$. Let $I \in \{0,1\}^{2^n \times 2^n}$ be the identity matrix and $J \in \{0,1\}^{2^n \times 2^n}$ be the all-$1$ matrix. Then
    \[ \HDone{n} = \sign\left(-\sum_{i \in [\log_2 n];\, j \in \{0,1\}} M_{i,j} - (k+1) I + (k + 1/2) J\right). \qedhere\]
\end{proof}

\subsection{Approximate Blocky Rank Upper Bound for $\HDone{n}$}
\label{sec:approx-br-bound}
In this section we adapt \cref{thm:sbr-upper-bound} to the setting of approximate blocky rank. 

We are going to use two simple technical claims for the error amplification.
\begin{lemma}
    \label{lem:blocky-rank-polynomial}
    Let $M \in \mathbb{R}^{n \times n}$ be a matrix and $p$ be a univariate polynomial over reals. Let $p(M) \in \mathbb{R}^{n \times n}$ be defined by $p(M)[x,y] \coloneqq p(M[x,y])$. We then have $\br(p(M)) \le \br(M)^{\deg(p)} \cdot \deg(p)$.
\end{lemma}
\begin{proof}
    It suffices to prove the lemma for $p = x^i$. Then for $M = \sum_{j \in [\br(M)]} \alpha_j B_j$ we have $p(M) = M^{\circ i}$ (where $M^{\circ i}$ is entrywise product of $i$ copies of $i$). Then we have $M^{\circ i} = \sum_{j_1, \dots, j_i \in [\br(M)]} \prod_{k \in [i]} \alpha_k \cdot B_{j_1} \circ \dots \circ B_{j_i}$. The proof is finished by observing that entrywise product of blocky matrices is a blocky matrix.
\end{proof}
The following claim directly follows from the  Chernoff-Hoeffding bound.
\begin{lemma}[see e.g. \cite{DGJSV10}]
    \label{lem:amplifying-poly}
    Let $A_k(x) \coloneqq \sum_{i=k/2}^k \binom{k}{i} x^i (1-x)^{k-i}$. Then for $x \in [0,1/2-\delta)$ we have $A_k(x) \le e^{-2k/\delta^2}$ and for $x \in (1/2+\delta, 1]$ we have $A_k(x) \ge 1 - e^{-2k/\delta^2}$.
\end{lemma}

\begin{theorem}
    $\br_\epsilon(\HDone{n}) = (\log n)^{O(\log (1/\epsilon))}$.
\end{theorem}
\begin{proof}
    As before, suppose that $n$ is a power of $2$ without loss of generality, otherwise increase $n$ so it is. Consider an error-correcting code $C \subseteq \{0,1\}^{K \log n}$ with $|C| = n$ and distance $(K \log_2 n) / 5$ (for example, an expander code with appropriate parameters suffices \cite{SS96}) i.e. for every $\alpha \neq \beta \in C$ we have $|\alpha \oplus \beta| \ge (K \log_2 n) / 5$.
    Associate each index $\ell \in [n]$ with a codeword $C_\ell$. Now for $i \in [K \log_2 n]$ and $j \in \{0,1\}$ define $M_{ij} 
    \in \{0,1\}^{2^n \times 2^n}$ as follows: $M_{ij}[x,y] = 0$ iff $x$ and $y$ coincide on all indices $\ell \in [n]$ such that $(C_t)_i = j$. 

    Now suppose $x$ and $y$ differ in exactly one position $\ell$: then for every $i \in [K \log_2 n]$ we have $M_{i,0}[x,y] + M_{i,1}[x,y] = 1$. 

    If $x$ and $y$ differ in at least two positions $\ell_1$ and $\ell_2$, then for every $i \in [K \log_2 n]$ we have $M_{i,0}[x,y] + M_{i,1}[x,y] \ge 1$ and for every $i$ such that $C_{\ell_1}$ and $C_{\ell_2}$ differ we get $M_{i,0}[x,y] + M_{i,1}[x,y] = 2$. 

    Thus, assuming $x \neq y$ we get $M'[x,y] \coloneqq \sum_{i \in [K \log_2 n]; j \in \{0,1\}} M_{ij}[x,y] \le K \log_2 n$ if $x$ and $y$ differ in one position and $M'[x,y] \ge K \log_2 n (1 + 1/5)$ if $x$ and $y$ differ in at least two positions. $M'[x,x] = 0$ for every $x \in \{0,1\}^n$.

    Let $I \in \{0,1\}^{2^n \times 2^n}$ be the identity matrix and $J \in \{0,1\}^{2^n \times 2^n}$ be the all-$1$ matrix. Then letting 
    \[M \coloneqq \frac{1}{2 K \log_2 n}\left(-M' - (K\log_2 n+1)I + ((2K-1/10) \log_2 n)J\right)\]
    we get 
    \( \left\| \HDone{n} - M\right\|_\infty \le 1/2-1/10K.\) 
    Now taking $M_{\text{final}} \coloneqq A_{200K^2 \log(1/\epsilon)}(M)$ by \cref{lem:amplifying-poly} and \cref{lem:blocky-rank-polynomial} we get the claim.
\end{proof}

\end{document}